\documentclass[11pt]{article}

\usepackage{fullpage}

\usepackage[utf8]{inputenc}
\usepackage[backend=bibtex,
maxbibnames=99,
style=numeric,
sorting=nyt,
firstinits=true
]{biblatex}

\usepackage{tikz}
\usetikzlibrary{math,calc,arrows, shapes,positioning,decorations.pathreplacing}

\usepackage[english]{babel}
\usepackage{amsmath,amsthm,amssymb,bm}
\usepackage{enumerate}

\newtheorem{theorem}{Theorem}
\newtheorem{lemma}{Lemma}
\newtheorem{corollary}{Corollary}

\newcommand{\abs}[1]{\vert #1 \vert}

\newcommand{\argmin}{\operatorname{argmin}}
\newcommand{\argmax}{\operatorname{argmax}}
\newcommand{\conv}{\operatorname{conv}}

\newenvironment{keywords}{\begin{trivlist}\item[]{\bfseries\sffamily Keywords:}\ }{\end{trivlist}}

\newenvironment{MSCclass}{\begin{trivlist}\item[]{\bfseries\sffamily MSC classes:}\ }{\end{trivlist}}
  
  \makeatletter

\newcommand*\samethanks[1][\value{footnote}]{\footnotemark[#1]}

\makeatother

\bibliography{ORprecsJournal.bib}

\begin{document}

\title{Approximation Algorithms and LP Relaxations for Scheduling Problems Related to Min-Sum Set Cover\thanks{This work has been supported by the Alexander von Humboldt Foundation with funds from the German Federal Ministry of Education and Research (BMBF). An extended abstract of an earlier version of this paper is scheduled to appear in \emph{Proceedings of WAOA 2019}.}}
\author{Felix Happach\thanks{ 
     Department of Mathematics and School of Management,
    Technische Universit\"{a}t M\"{u}nchen, Germany.\newline Emails: {\tt \{felix.happach,andreas.s.schulz\}@tum.de}}
    \and
    Andreas S.\ Schulz\samethanks[2]
}
\date{}

\maketitle

\begin{abstract}
We consider single-machine scheduling problems that are natural generalizations or variations of the min-sum set cover problem and the min-sum vertex cover problem. For each of these problems, we give new approximation algorithms. Some of these algorithms rely on time-indexed LP relaxations. We show how a variant of alpha-point scheduling leads to the best-known approximation ratios, including a guarantee of 4 for an interesting special case of the so-called generalized min-sum set cover problem. We also make explicit the connection between the greedy algorithm for min-sum set cover and the concept of Sidney decomposition for precedence-constrained single-machine scheduling, and show how this leads to a 4-approximation algorithm for single-machine scheduling with so-called bipartite OR-precedence constraints.  
\end{abstract}


\begin{MSCclass}
90B35; 90C05; 68W25
\end{MSCclass}

\section{Introduction}\label{sec:intro}

We consider the problem of scheduling jobs subject to AND/OR-precedence constraints on a single machine.
These scheduling problems are closely related to (precedence-constrained) min-sum set cover~\cite{FeigeLovaszTetali2002,FeigeLovaszTetali2004,McClintockMestreWirth2017} and generalized min-sum set cover~\cite{AzarGamzuYin2009,BansalGupta2010,SkutellaWilliamson2011,ImSviridenkoZwaan2014}.
Let $N = A \dot{\cup} B$ be the set of $n$ jobs with processing times $p_j \geq 0$ and weights $w_j \geq 0$ for all $j \in N$.
The precedence constraints are given by a directed graph $G=(N,E_{\wedge} \dot{\cup} E_{\vee})$, where $(i,j) \in E_{\wedge} \cup E_{\vee}$ means that job~$i$ is a predecessor of job~$j$.
The arcs in $E_{\wedge} \subseteq (A \times A) \cup (B \times B)$ and $E_{\vee} \subseteq A \times B$ represent AND- and OR-precedence constraints, respectively.
That is, a job in $N$ requires that \emph{all} its predecessors w.r.t.~$E_{\wedge}$ are completed before it can start.
A job in $B$, however, requires that \emph{at least one} of its predecessors w.r.t.~$E_{\vee}$ is completed beforehand.
The set of \emph{OR-predecessors} of job $b \in B$ is denoted by $\mathcal{P}(b) := \{a \in A \, | \, (a,b) \in E_{\vee} \}$.
Note that $\mathcal{P}(b)$ might be empty for some $b \in B$.

A \emph{schedule} $C$ is an ordering of the jobs on a single machine such that each job $j$ is processed non-preemptively for $p_j$ units of time, and no jobs overlap.
The \emph{completion time} of $j \in N$ in the schedule $C$ is denoted by $C_j$.
A schedule $C$ is \emph{feasible} if (i) $C_j \geq \max\{C_i\ | \ (i,j) \in E_{\wedge} \} + p_j$ for all $j \in N$ (AND-constraints), and (ii) $C_b \geq \min\{C_a \ | \ a \in \mathcal{P}(b) \} + p_b$ for all $b \in B$ with $\mathcal{P}(b) \not= \emptyset$ (OR-constraints).
The goal is to determine a feasible schedule $C$ that minimizes the sum of weighted completion times, $\sum_{j \in N} w_j C_j$.
We denote this problem by $1 \, | \, ao\text{-}prec=A \dot{\vee} B \, | \, \sum w_j C_j$, in an extension of the notation of Erlebach, K\"a\"ab and M\"ohring~\cite{ErlebachKaabMohring2003} and the three-field notation of Graham et al.~\cite{GrahamLawlerLenstraKan1979}. 
This scheduling problem is NP-hard. In fact, it generalizes a number of NP-hard problems, as discussed below.
Therefore, we focus on approximation algorithms.
Let $\Pi$ be a minimization problem, and $\rho \geq 1$.
Recall that a $\rho$-approximation algorithm for $\Pi$ is a polynomial-time algorithm that returns, for every instance of $\Pi$, a feasible solution with objective value at most $\rho$ times the optimal objective value.
If $\rho$ does not depend on the input parameters, we call the algorithm a \emph{constant-factor approximation}.

As already indicated, the scheduling problem we consider is motivated by its close connection to (min-sum) set covering problems.
Figure~\ref{fig:overview} gives an overview of related problems, which we describe briefly in the following paragraphs.

\paragraph*{Min-Sum Set Cover.}
The most basic problem is \emph{min-sum set cover} (MSSC), where the input consists of a hypergraph with vertices $V$ and hyperedges $\mathcal{E}$.
Given a linear ordering of the vertices $f: V \to \abs{V}$, the covering time of hyperedge $e \in \mathcal{E}$ is defined as $f(e) := \min_{v \in e} f(v)$.
The goal is to find a linear ordering of the vertices that minimizes the sum of covering times, $\sum_{e \in \mathcal{E}} f(e)$.
MSSC is indeed a special case of $1 \, | \, ao\text{-}prec=A \dot{\vee} B \, | \, \sum w_j C_j$:
we introduce a job in $A$ for every vertex of $V$ and a job in $B$ for every hyperedge in $\mathcal{E}$, and we set $p_a = w_b  = 1$ and $p_b = w_a = 0$ for all jobs $a \in A$ and $b \in B$.
Further, we let $E_{\wedge} = \emptyset$ and introduce an arc $(a,b) \in E_{\vee}$ in the precedence graph, if the vertex corresponding to~$a$ is contained in the hyperedge corresponding to~$b$. 

MSSC was first introduced by Feige, Lov\'asz and Tetali~\cite{FeigeLovaszTetali2002}, who observed that a simple greedy heuristic due to Bar-Noy et al.~\cite{Bar-NoyEtAl1998} yields an approximation factor of 4.
Feige et al.~\cite{FeigeLovaszTetali2002} simplified the analysis via a primal/dual approach based on a time-indexed linear program.
In the journal version of their paper, Feige et al.~\cite{FeigeLovaszTetali2004} also proved that it is NP-hard to obtain an approximation factor strictly better than 4.
The special case of MSSC where the hypergraph is an ordinary graph is called min-sum vertex cover (MSVC), and is APX-hard~\cite{FeigeLovaszTetali2004}.
Feige et al.~\cite{FeigeLovaszTetali2002,FeigeLovaszTetali2004} provided a 2-approximation for MSVC that is also based on a time-indexed linear program and uses randomized rounding.
Iwata, Tetali and Tripathi~\cite{IwataTetaliTripathi2012} improved the rounding scheme to obtain a 1.79-approximation for MSVC using the same linear program.

Munagala et al.~\cite{Munagala2005} generalized MSSC by introducing non-negative costs~$c_v$ for each vertex $v \in V$ and non-negative weights $w_e$ for each hyperedge $e \in \mathcal{E}$.
Here, the goal is to minimize the sum of weighted covering costs, $\sum_{e \in \mathcal{E}} w_e f(e)$, where the covering cost of $e \in \mathcal{E}$ is defined as $f(e):= \min_{v \in e} \sum_{w: f(w) \leq f(v)} c_w$.
The authors called this problem \emph{pipelined set cover} and proved, among other things, that the natural extension of the greedy algorithm of Feige et al.~for MSSC still yields a 4-approximation.
Similar to MSSC, one can model pipelined set cover as an instance of $1 \, | \, ao\text{-}prec=A \dot{\vee} B \, | \, \sum w_j C_j$.

Munagala et al.~\cite{Munagala2005} posed as an open problem whether there is still a constant-factor approximation for pipelined set cover, if there are AND-precedence constraints in form of a partial order~$\prec$ on the vertices of the hypergraph. That is, any feasible linear ordering $f: V \to \abs{V}$ must satisfy $f(v) < f(w)$, if $v \prec w$.
This question was partly settled by McClintock, Mestre and Wirth~\cite{McClintockMestreWirth2017}. They presented a $4\sqrt{\abs{V}}$-approximation algorithm for \emph{precedence-constrained MSSC}, which is the extension of MSSC where $E_{\wedge} = \{(a',a) \in A \times A \, | \, a' \prec a\}$. 
The algorithm uses a $\sqrt{\abs{V}}$-approximative greedy algorithm on a problem called max-density precedence-closed subfamily.
The authors also propose a reduction from the so-called planted dense subgraph conjecture~\cite{CharikarNaamadWirth2016} to precedence-constrained MSSC.
Roughly speaking, the conjecture says that for all $\varepsilon > 0$ there is no polynomial-time algorithm that can decide with advantage $> \varepsilon$ whether a random graph on $m$ vertices is drawn from $(m,m^{\alpha -1})$ or contains a subgraph drawn from $(\sqrt{m},\sqrt{m}^{\beta -1})$ for certain $0 < \alpha, \beta < 1$.\footnote{A random graph drawn from $(m,p)$ contains $m$ vertices and the probability of the existence of an edge between any two vertices is equal to $p$.}
If the conjecture holds true, then this implies that there is no $\mathcal{O}(\abs{V}^{1/12-\varepsilon})$-approximation for precedence-constrained MSSC~\cite{McClintockMestreWirth2017}.

The ordinary set cover problem~\cite{Karp1972} is also a special case of $1 \, | \, ao\text{-}prec=A \dot{\vee} B \, | \allowbreak \, \sum w_j C_j$: we can introduce a job in $A$ with $p_a = 1$ and $w_a = 0$ for every vertex of the hypergraph, a job in $B$ with $p_b = w_b = 0$ for every hyperedge, and an arc $(a,b) \in E_{\vee}$ in the precedence graph, if the vertex corresponding to $a$ is contained in the hyperedge corresponding to $b$.
Further, we include an additional job $x$ in $B$ with $p_x = 0$ and $w_x = 1$, and introduce an arc $(b,x) \in E_{\wedge}$ for every job $b \in B \setminus \{x\}$.
If the set cover instance admits a cover of cardinality $k$, we first schedule the corresponding vertex-jobs in $A$, so all hyperedge-jobs are available for processing at time~$k$.
Then job $x$ can complete at time $k$, which gives an overall objective value of~$k$.
Similarly, any schedule with objective value equal to $k$ implies that all hyperedge-jobs are completed before time~$k$, so there exists a cover of size at most~$k$.
Recall that set cover admits an $\ln(m)$-approximation~\cite{Johnson1974,Lovasz1975}, where $m$ is the number of hyperedges, and this is best possible, unless P=NP~\cite{DinurSteuer2014}.

 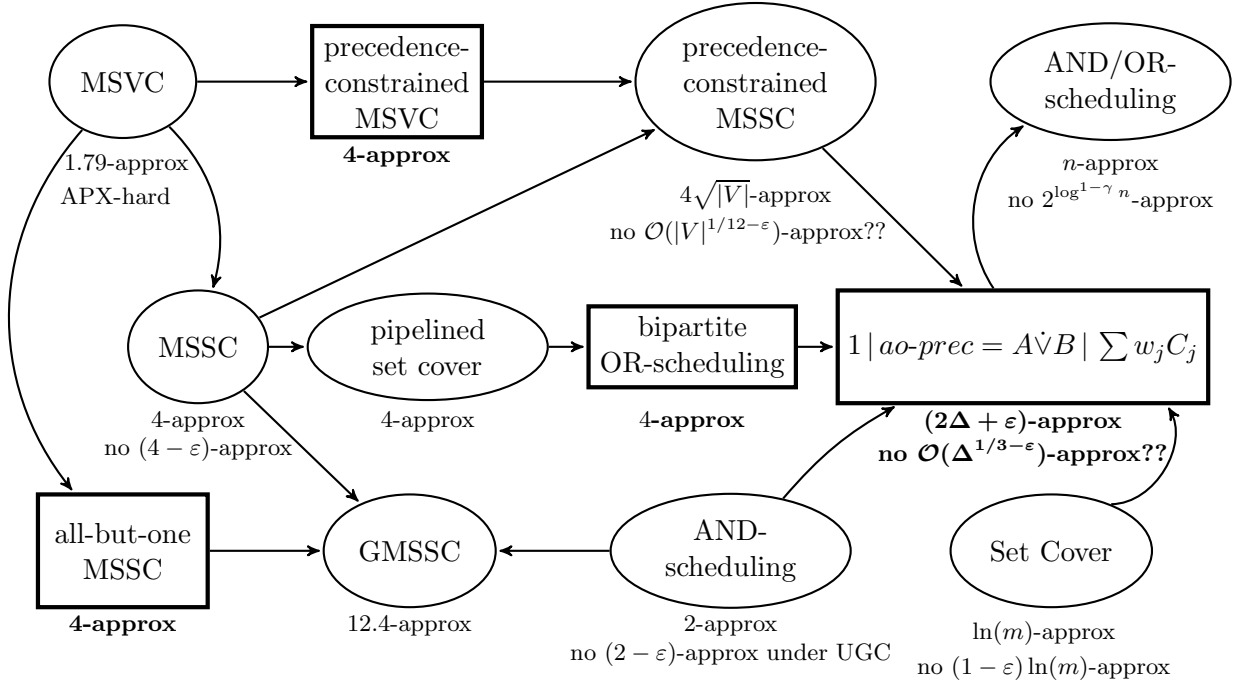
\begin{figure}
\centering
\begin{tikzpicture}[->,>=stealth',shorten >=1pt,auto,thick]

		\node[ellipse,draw,minimum size=1.5cm] (MSSC) {MSSC};
		\node [below of=MSSC] {{\footnotesize $4$-approx}};
		\node [below=0.3cm of MSSC] {{\footnotesize no $(4-\varepsilon)$-approx}};
		
		\node[ellipse,draw,text width=2cm,align=center] (pipeMSSC) [right=0.5cm of MSSC] {pipelined set cover};
		\node [below of=pipeMSSC] {{\footnotesize $4$-approx}};
		
		\node[rectangle,ultra thick,draw,text width=2.5cm,align=center] (bipOR) [right=0.5cm of pipeMSSC] {bipartite OR-scheduling};
		\node [below of=bipOR] {{\footnotesize \bf $4$-approx}};
		
		\node[rectangle,ultra thick,draw,minimum size=1.5cm] (scheduling) [right=0.5cm of bipOR] {$1 \, | \, ao\text{-}prec=A \dot{\vee} B \, | \, \sum w_j C_j$};
		\node [below of=scheduling] {{\footnotesize \bf $\bm{(2\Delta + \varepsilon)}$-approx}};
		\node [below=0.3cm of scheduling] {{\footnotesize \bf no $\bm{\mathcal{O}(\Delta^{1/3 - \varepsilon})}$-approx??}};

		\node[ellipse,draw,minimum size=1.5cm] (MSVC) [above=2cm of MSSC,xshift=-1cm] {MSVC};
		\node [below of=MSVC,yshift=-0.1cm,xshift=0.05cm] {{\footnotesize $1.79$-approx}};
		\node [below=0.5cm of MSVC,xshift=-0.1cm] {{\footnotesize APX-hard}};
		
		\node[rectangle,ultra thick,draw,text width=2cm,align=center] (precMSVC) [right=1.5cm of MSVC] {precedence-constrained MSVC};
		\node [below of=precMSVC] {{\footnotesize \bf $\bm{4}$-approx}};
		
		\node[ellipse,draw,text width=2cm,align=center] (precMSSC) [right=2cm of precMSVC] {precedence-constrained MSSC};
		\node [below=0.1cm of precMSSC] {{\footnotesize $4\sqrt{\abs{V}}$-approx}};
		\node [below=0.6cm of precMSSC,xshift=-0.1cm] {{\footnotesize no $\mathcal{O}(\abs{V}^{1/12- \varepsilon})$-approx??}};
		
		\node[ellipse,draw,minimum size=1cm,text width=2cm,align=center] (ANDOR) [right=1.5cm of precMSSC] {AND/OR-scheduling};
		\node [below=0.1cm of ANDOR] {{\footnotesize $n$-approx}};
		\node [below=0.4cm of ANDOR] {{\footnotesize no $2^{\log^{1-\gamma} n}$-approx}};

		\node[rectangle,ultra thick,draw,minimum size=1.5cm,text width=2cm,align=center] (aboMSSC) [below=4.7cm of MSVC] {all-but-one MSSC};
		\node [below of=aboMSSC] {{\footnotesize \bf $\bm{4}$-approx}};
		
		\node[ellipse,draw,minimum size=1.5cm] (GMSSC) [right=1.5cm of aboMSSC] {GMSSC};
		\node [below of=GMSSC] {{\footnotesize $12.4$-approx}};
		
		\node[ellipse,draw,text width=2cm,align=center] (AND) [right=1.5cm of GMSSC] {AND-scheduling};
		\node [below of=AND] {{\footnotesize $2$-approx}};
		\node [below=0.3cm of AND] {{\footnotesize no $(2-\varepsilon)$-approx under UGC}};		
		
		\node[ellipse,draw,minimum size=1.5cm] (setcover) [right=1.3cm of AND,align=center] {Set Cover};
		\node [below of=setcover,yshift=-0.1cm,xshift=-0.1cm] {{\footnotesize $\ln(m)$-approx}};
		\node [below=0.5cm of setcover,xshift=-0.1cm] {{\footnotesize no $(1-\varepsilon)\ln(m)$-approx}};
		
		\path (MSVC) edge (precMSVC)
				(precMSVC) edge (precMSSC)
				(precMSSC) edge (scheduling)
				(scheduling) edge[bend left=45] (ANDOR);
		\path (MSVC) edge [bend left=30] (MSSC)
				(MSSC) edge (pipeMSSC)
				(pipeMSSC) edge (bipOR)
				(bipOR) edge (scheduling);
		\path (MSSC) edge (precMSSC);
		\path (MSSC) edge (GMSSC);	
		\path (MSVC) edge[bend right=40] (aboMSSC)
				(aboMSSC) edge (GMSSC);	
		\path (AND) edge (GMSSC);
		\path (AND) edge [bend left=10] (scheduling);
		\path (setcover) edge [bend right=60] (scheduling);
		
\end{tikzpicture}
\caption{\small Overview of related problems and results. An arrow from problem $\Pi_1$ to $\Pi_2$ indicates that $\Pi_2$ generalizes $\Pi_1$. Problems in rectangular frames are explicitly considered in this paper, and our results are depicted in bold. Lower bounds indicated with ``??'' are assuming hardness of the planted dense subgraph problem~\cite{CharikarNaamadWirth2016}.}
\label{fig:overview}
\end{figure}

\paragraph*{New Approximation Algorithms.}
W.l.o.g., suppose that $E_{\wedge}$ is transitively closed, i.e.~$(i,j) \in E_{\wedge}$ and $(j,k) \in E_{\wedge}$ implies $(i,k) \in E_{\wedge}$.
We may further assume that there are no redundant OR-precedence constraints, i.e.~if $(a,b) \in E_{\vee}$ and $(a',a) \in E_{\wedge}$, then $(a',b) \notin E_{\vee}$.
Otherwise we could remove the arc $(a,b)$ from $E_{\vee}$, since any feasible schedule has to schedule $a'$ before $a$.
Let $\Delta := \max_{b \in B} \abs{\mathcal{P}(b)}$ be the maximum number of OR-predecessors of a job in $B$.
One can see that $\Delta$ is bounded from above by the cardinality of a maximum independent set in the induced subgraph on $E_{\wedge} \cap (A \times A)$.
Note that $\Delta$ is often relatively small compared to the total number of jobs.
For instance, if the precedence constraints are derived from an underlying graph, where the predecessors of each edge are its incident vertices (as in MSVC), then $\Delta \leq 2$.
Our first result is the following.
\begin{theorem}\label{thm:2DeltaApproximation}
There is a $2\Delta$-approximation algorithm for $1 \, | \, ao\text{-}prec=A \dot{\vee} B, \allowbreak \, p_j \in \{0,1\} \, |  \, \sum w_j C_j$.
Moreover, for any  $\varepsilon > 0$, there is a $(2\Delta + \varepsilon)$-approximation algorithm for $1 \, | \, ao\text{-}prec=A \dot{\vee} B \, | \, \sum w_j C_j$.
\end{theorem}

The proof of Theorem~\ref{thm:2DeltaApproximation} is contained in Section~\ref{sec:MSSC}.
First, we exhibit a randomized approximation algorithm for $1\, | \, ao\text{-}prec=A \dot{\vee} B, \, p_j \in \{0,1\} \, | \, \sum w_j C_j$, i.e.~if all processing times are 0/1, and then we show how to derandomize it.
This proves the first part of Theorem~\ref{thm:2DeltaApproximation}.

A natural question that arises in the context of real-world scheduling problems is whether approximation guarantees for 0/1-problems still hold for arbitrary processing times. 
As observed by Munagala et al.~\cite{Munagala2005}, the natural extension of the greedy algorithm for MSSC still works, if the processing times of jobs in $A$ are arbitrary, but all jobs in $B$ have zero processing time, and there are no AND-precedence constraints.
Once jobs in $B$ have non-zero processing times, their analysis of the greedy algorithm fails.
Our algorithm can be extended to arbitrary processing times (which proves the second part of Theorem~\ref{thm:2DeltaApproximation}) and, additionally, release dates.
In Section~\ref{sec:bipartite}, we also provide a 4-approximation algorithm for $1 \, | \, ao\text{-}prec=A \dot{\vee} B \, | \, \sum w_j C_j$ with $E_{\wedge} = \emptyset$, i.e., for single-machine scheduling with bipartite OR-precedence constraints.
Since $E_{\wedge} = \emptyset$, this problem is a special case of scheduling with OR-precedence constraints only, and we therefore denote it by $1 \, | \, or\text{-}prec=bipartite \, | \, \sum w_j C_j$.
The algorithm in Section~\ref{sec:bipartite} generalizes the greedy algorithm of~\cite{Munagala2005}, and is the first constant-factor approximation for bipartite OR-scheduling.

Note that the result of Theorem~\ref{thm:2DeltaApproximation} improves on the algorithm of~\cite{McClintockMestreWirth2017} for precedence-constrained MSSC in two ways.
First, the approximation factor of $2\Delta$ does not depend on the total number of jobs, but on the maximum number of OR-predecessors of a job in $B$.
In particular, we immediately obtain a 4-approximation for the special case of precedence-constrained MSVC.
Secondly, the algorithm works for arbitrary processing times, additional AND-precedence constraints on $B \times B$, and it can be extended to non-trivial release dates $r_j \geq 0$ of the jobs.
Note that, in general, $\Delta$ and $\sqrt{\abs{V}}$ are incomparable.
In most practically relevant instances, $\Delta$ should be considerably smaller than $\sqrt{\abs{V}}$.

It is important to highlight that the approximation factor of $2 \Delta$ in Theorem~\ref{thm:2DeltaApproximation} does not contradict the conjectured hardness of precedence-constrained MSSC stated in~\cite{McClintockMestreWirth2017}.
The set $A$ in the reduction of~\cite{McClintockMestreWirth2017} from the planted dense subgraph problem contains a job for every vertex and every edge of the random graph on $m$ vertices.
Each vertex-job consists of the singleton $\{0\}$ whereas each edge-job is a (random) subset of $[q] := \{1,\dots,q\}$, for some non-negative integer $q$.
Every element in $[q]$ appears in expectation in $mp^2$ many edge-jobs, where $p$ is a carefully chosen probability.
If we interpret this as a scheduling problem, we can delete the dummy element $0$ from the instance.
So the maximum indegree of a job in $B =[q]$ (maximum number of appearances of the element) is $\Delta \approx m p^2 \geq m^{\frac{1}{4}}$, see~\cite{McClintockMestreWirth2017}.
Hence the gap $\Omega(m^{\frac{1}{8}})$ in the reduction translates to a gap of $\Omega(\sqrt{\Delta})$ in our setting.
Therefore, if the planted dense subgraph conjecture~\cite{CharikarNaamadWirth2016} holds true, then there is no $\mathcal{O}(\Delta^{1/3 - \varepsilon})$-approximation algorithm for $1 \, | \, ao\text{-}prec=A \dot{\vee} B \, | \, \sum w_j C_j$ for any $\varepsilon > 0$.

Note that in the reduction from set cover to $1 \, | \, ao\text{-}prec=A \dot{\vee} B \, | \, \sum w_j C_j$ the parameter $\Delta$ equals the maximum cardinality of any hyperedge in the set cover instance.
Hochbaum~\cite{Hochbaum1982} presented an approximation algorithm for set cover with a guarantee of~$\Delta$.
Hence, the $2\Delta$-approximation of Theorem~\ref{thm:2DeltaApproximation} does not contradict the hardness of obtaining a $(1-\varepsilon)\ln(m)$-approximation for set cover~\cite{DinurSteuer2014}.
If the planted dense subgraph conjecture~\cite{CharikarNaamadWirth2016} is false, then constant-factor approximations for $1 \, | \, ao\text{-}prec=A \dot{\vee} B \, | \, \sum w_j C_j$ with $E_{\wedge} \subseteq A \times A$ may be possible.
However, the reduction from set cover shows that, in general, we cannot get a constant-factor approximation if $E_{\wedge} \cap (B \times B) \not= \emptyset$.

\paragraph*{Generalized Min-Sum Set Cover.}
A different generalization of MSSC, called \emph{generalized min-sum set cover} (GMSSC), was introduced by Azar, Gamzu and Yin~\cite{AzarGamzuYin2009}.
The input of GMSSC is similar to MSSC, but, in addition, each hyperedge $e \in \mathcal{E}$ is associated with a covering requirement $\kappa(e) \in [\abs{e}]$, where $\abs{e}$ is the cardinality of hyperedge $e$.
Given a linear ordering of the vertices, the covering time of $e \in \mathcal{E}$ is now the first point in time when $\kappa(e)$ of its incident vertices appear in the linear ordering.
The goal is again to minimize the sum of covering times over all hyperedges.

In our notation, this means that $E_{\wedge} = \emptyset$ and each job $b \in B$ requires at least $\kappa(b) \in [\abs{\mathcal{P}(b)}]$ of its OR-predecessors to be completed before it can start.
The extreme cases $\kappa(b) = 1$ and $\kappa(b) = \abs{\mathcal{P}(b)}$ are MSSC and the minimum latency set cover problem, respectively.
The latter is, in fact, equivalent to single-machine scheduling with AND-precedence constraints~\cite{Woeginger2003}.
Over time, several constant-factor approximations for GMSSC were proposed.
Bansal, Gupta and Krishnaswamy~\cite{BansalGupta2010} presented an algorithm with an approximation guarantee of 485, which was improved to 28 by Skutella and Williamson~\cite{SkutellaWilliamson2011}.
Both algorithms are based on the same time-indexed linear program, but use different rounding techniques, namely standard randomized rounding~\cite{BansalGupta2010} and $\alpha$-points~\cite{SkutellaWilliamson2011}, respectively.

The currently best-known approximation ratio for GMSSC is 12.4, due to Im, Sviridenko and Zwaan~\cite{ImSviridenkoZwaan2014}.
However, Im et al.~\cite{ImSviridenkoZwaan2014} conjecture that GMSSC admits a 4-approximation. 
By adapting the proof of Theorem~\ref{thm:2DeltaApproximation}, we obtain a 4-approximation for GMSSC if $\kappa(b) = \max\{\abs{\mathcal{P}(b)} - 1,1\}$, for all $b \in B$.
To the best of the authors' knowledge, this case, which we call \emph{all-but-one MSSC}, was not considered before.
Here, each job (with more than one predecessor) needs at least all but one of them to be completed before it can start.
This is a natural special case inbetween MSSC and AND-precedence constrained scheduling (where $\kappa(b)= 1$ and $\kappa(b) = \abs{\mathcal{P}(b)}$, respectively).
Note that all-but-one MSSC generalizes MSVC.
The proof of Theorem~\ref{thm:4approximation:allbutoneMSSC} below is contained in Section~\ref{sec:GMSSC}.

\begin{theorem}\label{thm:4approximation:allbutoneMSSC}
There is a 4-approximation algorithm for all-but-one MSSC.
\end{theorem}

\paragraph*{Related Work on Scheduling Problems.}
The first polynomial-time algorithm for scheduling jobs on a single machine to minimize the sum of weighted completion times is due to Smith~\cite{Smith1956}.
Once there are AND-precedence constraints, the problem becomes strongly NP-hard~\cite{LenstraKan1978}.
The first constant-factor approximation for AND-precedence constraints was proposed by Hall, Shmoys and Wein~\cite{HallShmoysWein1996}. It had an approximation factor of $4 + \varepsilon$.
Their algorithm is based on a time-indexed linear program and $\alpha$-point scheduling, but with a fixed value of $\alpha$.
Subsequently, various $2$-approximations based on linear programs~\cite{Schulz1996,HallSchulz1997,ChudakHochbaum1999} as well as purely combinatorial algorithms~\cite{ChekuriMotwani1999,MargotQueyranneWang2003} were derived.
Assuming a variant of the Unique Games Conjecture of Khot~\cite{Khot2002}, Bansal and Khot~\cite{BansalKhot2009} showed that the approximation ratio of $2$ is essentially best possible.

If the precedence constraints are of AND/OR-structure, then the problem does not admit constant-factor approximations anymore.
Let $0< c < \frac{1}{2}$ and $\gamma = (\log \log n )^{-c}$.
It is NP-hard to approximate the sum of weighted completion times of unit processing time jobs on a single machine within a factor of $2^{\log^{1-\gamma} n}$, if AND/OR-precedence constraints are involved~\cite{ErlebachKaabMohring2003}.
The precedence graph in the reduction consists of four layers with an OR/AND/OR/AND-structure.
Erlebach, K\"a\"ab and M\"ohring~\cite{ErlebachKaabMohring2003} also showed that scheduling the jobs in order of non-decreasing processing times (among the available jobs) yields an $n$-approximation for general weights and a $\sqrt{n}$-approximation for unit weights, respectively.
It can easily be verified that $1 \, | \, ao\text{-}prec=A \dot{\vee} B \, | \, \sum w_j C_j$ is a special case of the problem considered in~\cite{ErlebachKaabMohring2003}.

Scheduling unit processing time jobs with OR-precedence constraints only on parallel machines to minimize the sum of completion times can be solved in polynomial time~\cite{Johannes2005}.
However, once we want to minimize the sum of \emph{weighted} completion times, already the single-machine problem with unit processing times becomes strongly NP-hard~\cite{Johannes2005}.
In Section~\ref{sec:NPhardness}, we extend this result by showing that the problem remains NP-hard, even if we restrict the weights to be 0/1.

\paragraph*{Our Techniques and LP Relaxations.}
The algorithms that lead to Theorems~\ref{thm:2DeltaApproximation} and~\ref{thm:4approximation:allbutoneMSSC} are based on time-indexed linear programs and the concept of random $\alpha$-point scheduling, similar to, e.g.,~\cite{GoemansPersonal,HallShmoysWein1996,HallSchulz1997,SchulzSkutella1997,ChekuriSteinEtal2001,GoemansEtAl2002}.
One new element here is to not use a global value for $\alpha$, but to use different values of $\alpha$ for the jobs in $A$ and $B$, respectively.
This is crucial in order to obtain feasible schedules.
We focus on time-indexed linear programs, since other standard LP formulations fail in the presence of OR-precedence constraints.

More specifically, we show in Section~\ref{sec:otherformulations} that these relaxations have an integrality gap that is linear in the number of jobs, even on instances with $\Delta = 2$ and $E_{\wedge} = \emptyset$.
In Section~\ref{sec:linearordering}, we discuss a formulation in linear ordering variables that was introduced by Potts~\cite{Potts1980}.
We present a class of constraints that is facet-defining for the integer hull (Theorem~\ref{thm:ORfacetdefining}), and prove that the integrality gap remains linear, even if we add these inequalities.
In Section~\ref{sec:completiontime}, we consider an LP relaxation in completion time variables, which was proposed by Wolsey~\cite{WolseyISMPTalk} and Queyranne~\cite{Queyranne1993}.
We first generalize the well-known parallel inequalities~\cite{WolseyISMPTalk,Queyranne1993}, which fully describe the polytope in the absence of precedence constraints, to OR-precedence constraints (Theorem~\ref{thm:singlemachinepolytope:validrelaxation}).
Then we show that, even though we add an exponential number of tight valid inequalities, the corresponding LP relaxation still exhibits a linear integrality gap.

\section{A New Generalization of Min-Sum Set Cover}\label{sec:MSSC}

Consider an instance of $1\, | \, ao\text{-}prec=A \dot{\vee} B \, | \, \sum w_j C_j$. W.l.o.g., we may assume that $w_a = 0$ for all $a \in A$.
Otherwise, we can shift a positive weight of a job in $A$ to an additional successor in $B$ with zero processing time.
Further, we may assume that all data is integer and $p_j \geq 1$ for every job $j \in N$ that has no predecessors (otherwise such a job can be disregarded).
So no job can complete at time 0 in a feasible schedule.

Suppose that $p_j \in \{0,1\}$ for all $j \in N$, and let $T = \sum_{j \in N} p_j$ be the time horizon.
We consider the time-indexed linear programming formulation of Sousa and Wolsey~\cite{SousaWolsey1992} with AND-precedence constraints~\cite{HallShmoysWein1996}.
The binary variable $x_{jt}$ indicates whether job $j \in N$ completes at time $t \in [T]$ or not.
Additionally, we introduce constraints corresponding to $E_{\vee}$.
The resulting linear relaxation is
\begin{subequations}  \label{ANDprecsbipartiteunitp:LP}
\begin{align}
& \min & \sum\limits_{b \in B} \sum\limits_{t = 1}^{T} w_b \cdot t \cdot x_{bt}  \\
& \text{s.t.} & \sum\limits_{t=1}^T x_{jt} &= 1 & \forall &\,j \in N, \label{ANDprecsbipartiteunitp:LP:execute} \\
				&&  \sum\limits_{j \in N} \sum\limits_{s=t -p_j + 1}^{t} x_{js} &\leq 1  & \forall &\, t \in [T], \label{ANDprecsbipartiteunitp:LP:nooverlap} \\
 				&& \sum\limits_{s=1}^{t + p_b} x_{bs} - \sum\limits_{a \in \mathcal{P}(b)} \sum\limits_{s=1}^{t}  x_{as} &\leq 0 & \forall &\, b \in B: \, \mathcal{P}(b) \not= \emptyset, \ \forall \, t \in [T -p_b], \label{ANDprecsbipartiteunitp:LP:OR} \\
 				&& \sum\limits_{s =1}^{t+p_j} x_{js} - \sum\limits_{s=1}^t x_{is} &\leq 0  &\forall &\, (i,j) \in E_{\wedge}, \ \forall \, t \in [T - p_j], \label{ANDprecsbipartiteunitp:LP:AND}\\
 			&& x_{jt} 	&\geq 0  & \forall &\, j \in N, \ \forall \, t \in [T] \label{ANDprecsbipartiteunitp:LP:relax}.
\end{align} 
\end{subequations}
Constraints~(\ref{ANDprecsbipartiteunitp:LP:execute}) and~(\ref{ANDprecsbipartiteunitp:LP:nooverlap}) ensure that each job is executed and no jobs overlap, respectively.
Note that only jobs with $p_j = 1$ appear in~(\ref{ANDprecsbipartiteunitp:LP:nooverlap}).
Constraints~(\ref{ANDprecsbipartiteunitp:LP:OR}) and~(\ref{ANDprecsbipartiteunitp:LP:AND}) ensure OR- and AND-precedence constraints, respectively.
Note that we can solve LP~(\ref{ANDprecsbipartiteunitp:LP}) in polynomial time, since $T \leq n$. If $\mathcal{P}(b) = \emptyset$ for all $b \in B$ then the instance is an instance of scheduling with AND-precedence constraints only. In this case, we set $\Delta = 1$ in the following.

Let $\overline{x}$ be an optimal fractional solution of LP~(\ref{ANDprecsbipartiteunitp:LP}).
For $j \in N$, we call $\overline{C}_j = \sum_t t \cdot \overline{x}_{jt}$ its \emph{fractional completion time}.
Note that $\sum_j w_j \overline{C}_j$ is a lower bound on the objective value of an optimal integer solution, which corresponds to an optimal schedule.
For $0 < \alpha \leq 1$ and $j \in N$, we define its $\alpha$-point, $t^\alpha_j := \min\{t \, | \, \sum_{s=1}^t \overline{x}_{js} \geq \alpha\}$, to be the first integer point in time when an $\alpha$-fraction of $j$ is completed~\cite{HallShmoysWein1996}.

The algorithm, hereafter called Algorithm 1, works as follows. 
First, solve LP~(\ref{ANDprecsbipartiteunitp:LP}) to optimality, and let $\overline{x}$ be an optimal fractional solution.
Then, draw $\beta$ at random from the interval $(0,1]$ with density function $f(\beta) = 2\beta$, and set $\alpha = \frac{\beta}{\Delta}$.
(Choosing $\alpha$ as a function of $\beta$ is crucial in order to obtain a feasible schedule in the end. This together with~(\ref{ANDprecsbipartiteunitp:LP:OR}) ensures that at least one OR-predecessor of a job $b \in B$ completes early enough in the constructed schedule. The density function $f(\beta) = 2\beta$ is chosen to cancel out an unbounded term of $\frac{1}{\beta}$ in the expected value of the completion time of job $b$, as in~\cite{GoemansPersonal,SchulzSkutella1997}.)
Now, compute $t^\alpha_a$ and $t^\beta_b$ for all jobs $a \in A$ and $b \in B$, respectively.
Sort the jobs in order of non-decreasing values $t^\alpha_a$ ($a \in A$) and $t^\beta_b$ ($b \in B$), and denote this total order by $\prec$.
If there is $b \in B$ and $a \in \mathcal{P}(b)$ with $t^\alpha_a = t^\beta_b$, then set $a \prec b$.
Similarly, set $i \prec j$, if $(i,j) \in E_{\wedge}$ and $t^\alpha_{i} = t^\alpha_j$ (for $i,j \in A$) or $t^\beta_{i} = t^\beta_j$ (for $i,j \in B$).
(Recall that $E_{\wedge} \subseteq (A \times A) \cup (B \times B)$, so $(i,j) \in E_{\wedge}$ implies $i,j \in A$ or $i,j \in B$.)
Break all other ties arbitrarily.
Our main result shows that ordering jobs according to $\prec$ yields a feasible schedule and that the expected objective value of this schedule is at most $2\Delta$ times the optimum.

\begin{lemma}\label{lem:ANDprecsbipartiteunitp}
Algorithm 1 is a randomized $2 \Delta$-approximation for $1 \, | \, ao\text{-}prec=A \dot{\vee} B, \, p_j \in \{0,1\} \, |\allowbreak \, \sum w_j C_j$.
\end{lemma}

\begin{proof}
Note that Algorithm 1 runs in polynomial time, since we can solve LP~(\ref{ANDprecsbipartiteunitp:LP}) in polynomial time.
We first show that scheduling the jobs in order of $\prec$ yields a feasible schedule for any fixed $0 < \beta \leq 1$.
Recall the definition of $\Delta = \max_{b \in B} \abs{\mathcal{P}(b)}$.
Let $0 < \beta \leq 1$ and set $\alpha = \alpha(\beta) = \frac{\beta}{\Delta}$.

Note that $t^{\alpha}_{i} \leq t^{\alpha}_j$ for any $(i,j) \in E_{\wedge} \cap (A \times A)$ and $t^{\beta}_{i} \leq t^{\beta}_j$ for any $(i,j) \in E_{\wedge} \cap (B \times B)$, by~(\ref{ANDprecsbipartiteunitp:LP:AND}).
Due to the tie breaking rule, we have $i \prec j$ whenever $(i,j) \in E_{\wedge}$.
For $b \in B$ with $\mathcal{P}(b) \not= \emptyset$, constraint~(\ref{ANDprecsbipartiteunitp:LP:OR}) implies
$\beta \leq \sum_{s=1}^{t^\beta_b} \overline{x}_{bs} \leq \sum_{a \in \mathcal{P}(b)} \sum_{s=1}^{t^\beta_b - p_b} \overline{x}_{as}$.
So there is $a_b \in \mathcal{P}(b)$ such that $\sum_{s=1}^{t^\beta_b -p_b} \overline{x}_{a_b s} \geq \frac{\beta}{\abs{\mathcal{P}(b)}} \geq \alpha$.
Hence, $t^{\alpha}_{a_b} \leq t^\beta_b$, and thus $a_b \prec b$, by the tie breaking rule.
So, $\prec$ satisfies all prececende constraints, and Algorithm 1 returns a feasible schedule.

As for the approximation factor, fix $j \in N$.
For $t \in \{0,\dots,T\}$, let $\alpha_t = \sum_{s=1}^t \overline{x}_{js}$ be the fraction of job $j$ that is completed by time $t$.
Note that $\alpha_0 = 0$, $\alpha_T = 1$ and $t^\gamma_j \leq t$ if $\gamma \leq \alpha_t$.
For $0 < \gamma \leq 1$ and $j \in N$, we observe, similar to~\cite{Goemans1997,SchulzSkutella1997}, that
\begin{equation}\label{ineq:integral:t}
\int_0^1 t^\gamma_j d\gamma = \sum_{t=1}^{T} \int_{\alpha_{t-1}}^{\alpha_{t}} t^\gamma_j d \gamma \leq \sum_{t=1}^{T} (\alpha_{t} - \alpha_{t-1}) t = \sum_{t=1}^{T} \left( \sum_{s=1}^{t} \overline{x}_{js} - \sum_{s=1}^{t-1} \overline{x}_{js} \right) t = \sum_{t=1}^T \overline{x}_{jt} \cdot t = \overline{C}_j.
\end{equation}

For fixed $0 < \beta \leq 1$ and $i, j \in N$, let $\eta_i^j(\beta) = \sum_{s = 1}^{t^\beta_j} \overline{x}_{is}$ be the fraction of $i$ that is processed before $t^\beta_j$.
Note that there is no idle time on the machine in the optimal fractional solution $\overline{x}$, so $\sum\limits_{i \in N} \eta^j_i(\beta) p_i =t^\beta_j $.
Let $b \in B$ and $i \in N$ with $i \prec b$.
The completion times of the schedule returned by Algorithm 1 for a specific realization of $\beta$ are denoted by $C_j(\beta)$.
By construction we have $\alpha \leq \eta^b_i (\beta)$ (if $i \in A$) and $\alpha \leq \beta \leq \eta^b_i(\beta)$ (if $i \in B$), respectively.
So
\begin{equation}\label{ineq:completion:t}
\alpha \, C_b(\beta) = \sum_{i \preceq b} \alpha \, p_i \leq \sum_{i \preceq b} \eta^b_i(\beta) \, p_i \leq t^\beta_b.
\end{equation}
Thus, the expected completion time of $b \in B$ is
\begin{equation}\label{ineq:expectation:t}
\mathbb{E}[C_b(\beta)] = \int_0^1 f(\beta) C_b(\beta) d \beta \leq \int_0^1 f(\beta) \frac{\Delta}{\beta} t^\beta_b d \beta = 2 \Delta \int_0^1 t^\beta_b d \beta \leq 2 \Delta \overline{C}_b.
\end{equation}
Since only jobs in $B$ contribute to the objective function, this yields the claim.
\end{proof}

For fixed $\overline{x}$ and $0 < \beta \leq 1$ we call the schedule that orders the jobs according to $\prec$ the \emph{$\beta$-schedule of $\overline{x}$}.
Given $\overline{x}$ and $0 < \beta \leq 1$, we can construct the $\beta$-schedule in time $\mathcal{O}(n)$.
We derandomize Algorithm 1 by a simple observation similar to~\cite{ChekuriSteinEtal2001,GoemansEtAl2002}.
List all possible schedules that occur as $\beta$ goes from 0 to 1, and pick the best one.
The next lemma shows that the number of different $\beta$-schedules is not too large.

\begin{lemma}\label{lem:ANDprecsbipartite:derandomize}
For every $\overline{x}$ there are $\mathcal{O}(n^2)$ different $\beta$-schedules. 
\end{lemma}
\begin{proof}
Note that, as $\beta$ goes from 0 to 1, $\alpha(\beta) = \frac{\beta}{\Delta}$ is a linear function with values from $0$ to~$\frac{1}{\Delta}$.
We interpret the fractional solution $\overline{x}$ as a preemptive schedule where an $\overline{x}_{jt}$-fraction of job $j$ is contiguously scheduled in time slot $[t-1;t]$.
Thus the order of jobs in the $\beta$-schedule of $\overline{x}$ only changes, if the $\alpha(\beta)$-point or $\beta$-point of a job in $A$ or $B$ reaches a point when this job gets preempted, respectively. 
So the number of different $\beta$-schedules is bounded from above by the number of preemptions in $\overline{x}$.
Each job is preempted at most once within each time step.
Recall that $T \in \mathcal{O}(n)$, since $p_j \in \{0,1\}$ for all $j \in N$.
So there are at most $n \cdot T \in \mathcal{O}(n^2)$ preemptions.
\end{proof}

Lemmas~\ref{lem:ANDprecsbipartiteunitp} and~\ref{lem:ANDprecsbipartite:derandomize} together prove the first part of Theorem~\ref{thm:2DeltaApproximation}.
Note that for scheduling instances that are equivalent to MSVC, $\Delta \leq 2$. Hence, we immediately obtain a 4-approximation for these instances.
\begin{corollary}\label{cor:ANDprecsMSVC}
There is a 4-approximation algorithm for precedence-constrained MSVC.
\end{corollary}

If we use an interval-indexed LP instead of a time-indexed LP (see also~\cite{HallShmoysWein1996,HallSchulz1997}), then Algorithm~1 can be generalized to arbitrary processing times. This will prove the second part of Theorem~\ref{thm:2DeltaApproximation}.
Let $\varepsilon' > 0$, and recall that all processing times are non-negative integers.
Let $T = \sum_{j \in N} p_j$ be the time horizon, and let $L$ be minimal such that $(1+\varepsilon')^{L-1} \geq T$.
Set $\tau_{0} := 1$, and let $\tau_l = (1+ \varepsilon')^{l-1}$ for every $l \in [L]$.
We call $(\tau_{l-1},\tau_l]$ the $l$-th interval for $l \in [L]$. 
(The first interval is the singleton $(1,1]:=\{1\}$.)
We introduce a binary variable $x_{jl}$ for every $j \in N$ and for every $l \in [L]$ that indicates whether or not job $j$ completes in the $l$-th interval.
If we relax the integrality constraints on the variables we obtain the following relaxation:
\begin{subequations}\label{ANDprecsMSSC:generalprocessing:LP}
\begin{align}
&\min & \sum\limits_{b \in B} \sum\limits_{l = 1}^{L} w_b \cdot \tau_{l-1} \cdot x_{bl} \\
 &\text{s.t.} & \sum\limits_{l=1}^L x_{jl} &= 1  & \forall &\,j \in N, \label{ANDprecsMSSC:generalprocessing:LP:execute}\\
 				&& \sum\limits_{j \in N} \sum\limits_{k=1}^l  p_j \, x_{jk} &\leq \tau_l  & \forall &\, l \in [L], \label{ANDprecsMSSC:generalprocessing:LP:nooverlap}\\
 				&& \sum\limits_{k=1}^l x_{bk} - \sum\limits_{a \in \mathcal{P}(b)} \sum\limits_{k=1}^l x_{ak} &\leq 0 & \forall &\, b \in B: \, \mathcal{P}(b) \not= \emptyset, \ \forall \, l \in [L], \label{ANDprecsMSSC:generalprocessing:LP:OR} \\
 				&& \sum\limits_{k =1}^{l} x_{jk} - \sum\limits_{k=1}^l x_{ik} &\leq 0  &\forall &\, (i,j) \in E_{\wedge}, \ \forall \, l \in [L], \label{ANDprecsMSSC:generalprocessing:LP:AND} \\
 				&& x_{jl} 	&\geq 0 & \forall &\, j \in N, \ \forall \, l \in [L]: \, \tau_{l-1} \geq p_j \label{ANDprecsMSSC:generalprocessing:LP:relax}.
\end{align}
\end{subequations}
Given $\varepsilon'$, the size of LP~(\ref{ANDprecsMSSC:generalprocessing:LP}) is polynomial, so we can solve it in polynomial time.
Again~(\ref{ANDprecsMSSC:generalprocessing:LP:execute}) ensures that every job is executed.
Constraints~(\ref{ANDprecsMSSC:generalprocessing:LP:nooverlap}) are valid for any feasible schedule, since the total processing time of all jobs that complete within the first $l$ intervals cannot exceed $\tau_l$. 
Constraints~(\ref{ANDprecsMSSC:generalprocessing:LP:OR}) and~(\ref{ANDprecsMSSC:generalprocessing:LP:AND}) ensure that, at the end of each interval, the fractions of the jobs satisfy OR- and AND-precedence constraints, respectively.

Let $\overline{x}$ be an optimal fractional solution of LP~(\ref{ANDprecsMSSC:generalprocessing:LP}), and let $\overline{C}_j = \sum_l \tau_{l-1} \, \overline{x}_{jl}$.
Note that $\sum_j w_j \, \overline{C}_j$ is a lower bound on the optimal objective value of an integer solution, which is a lower bound on the optimal value of a feasible schedule.
Let $l^\alpha_j = \min\{l \, | \, \sum_{k=1}^l \overline{x}_{jk} \geq \alpha\}$ be the $\alpha$-interval of job $j \in N$.
This generalizes the notion of $\alpha$-points from before.

The algorithm for arbitrary processing times is similar to Algorithm 1.
We call it Algorithm 2 and it works as follows.
In order to achieve a $(2 \Delta + \varepsilon)$-approximation, solve LP~(\ref{ANDprecsMSSC:generalprocessing:LP}) with $\varepsilon' = \frac{\varepsilon}{2\Delta}$ and let $\overline{x}$ be an optimal solution.
Then, draw $\beta$ at random from the interval $(0,1]$ with density function $f(\beta) = 2\beta$, and set $\alpha = \frac{\beta}{\Delta}$.
Compute $l^\alpha_a$ and $l^\beta_b$ for all jobs $a \in A$ and $b \in B$, respectively.
Sort the jobs in order of non-decreasing values $l^\alpha_a$ ($a \in A$) and $l^\beta_b$ ($b \in B$) and denote this total order by $\prec$.
If $l^\alpha_a = l^\beta_b$ for some $b \in B$ and $a \in \mathcal{P}(b)$, set $a \prec b$.
Similarly, set $i \prec j$, if $(i,j) \in E_{\wedge}$ and $l^\alpha_{i} = l^\alpha_j$ (for $i,j \in A$) or $l^\beta_{i} = l^\beta_j$ (for $i,j \in B$).
Break all other ties arbitrarily.
Finally, schedule the jobs in the order of $\prec$.
Note that $\prec$ extends the order for $\alpha$-points from Algorithm 1 to $\alpha$-intervals.

\begin{lemma}\label{lem:ANDprecsMSSC:generalprocessingtime}
Algorithm 2 is a randomized $(2 \Delta + \varepsilon)$-approximation for $1 \, | \, ao\text{-}prec=A \dot{\vee} B \, | \, \sum w_j C_j$, for any $\varepsilon > 0$.
\end{lemma}
\begin{proof}
The proof is similar to the proof of Lemma~\ref{lem:ANDprecsbipartiteunitp}, see Appendix~\ref{appendix:MSSC}.
\end{proof}

We can derandomize Algorithm 2 similar to Lemma~\ref{lem:ANDprecsbipartite:derandomize}.
Interpret $\overline{x}$ as a preemptive schedule that assigns jobs to intervals.
Note that each job is preempted at most once per time interval $(\tau_{l-1},\tau_l]$.
So the number of $\beta$-schedules is bounded from above by $n \cdot L$ which is polynomially bounded in the input size.
This proves the second part of Theorem~\ref{thm:2DeltaApproximation}.

Algorithms 1 and 2 can be further extended to release dates.
To do so, we need to add constraints to LP~(\ref{ANDprecsbipartiteunitp:LP}) and LP~(\ref{ANDprecsMSSC:generalprocessing:LP}) that ensure that no job completes too early.
More precisly, fix $x_{jt} = 0$ for all $j \in N$ and $t < r_j + p_j$ in LP~(\ref{ANDprecsbipartiteunitp:LP}) and $x_{jl} = 0$ for all $j \in N$ and $\tau_{l-1} < r_j + p_j$ in LP~(\ref{ANDprecsMSSC:generalprocessing:LP}), respectively.
When scheduling the jobs according to $\prec$, we might have to add idle time in order to respect the release dates.
This increases the approximation factor slightly.

\begin{lemma}\label{lem:ANDprecsMSSC:releasedates}
There is a $(2\Delta + 2)$- and $(2\Delta + 2 + \varepsilon)$-approximation algorithm for $1 \, | \, r_j, \, ao\text{-}prec=A \dot{\vee} B, \allowbreak \, p_j \in \{0,1\} \, | \, \sum w_j C_j$ and $1 \, | \, r_j, \, ao\text{-}prec=A \dot{\vee} B \, | \, \sum w_j C_j$, respectively.
\end{lemma}
\begin{proof}
The proof works similar to the proofs of Lemma~\ref{lem:ANDprecsbipartiteunitp} and Lemma~\ref{lem:ANDprecsMSSC:generalprocessingtime}, see Appendix~\ref{appendix:MSSC}. 
\end{proof}

\section{The Generalized Min-Sum Set Cover Problem}\label{sec:GMSSC}

Recall that, in GMSSC, we are given a hypergraph where each hyperedge $e \in \mathcal{E}$ has a certain covering requirement $\kappa(e) \in [\abs{e}]$. Given a linear ordering of the vertices, a hyperedge $e$ is covered as soon as $\kappa(e)$ of its incident vertices appeared in the linear ordering. The goal is to find a linear ordering that minimizes the sum of covering times over all hyperedges. We can model GMSSC as a single-machine scheduling problem to minimize the sum of weighted completion times with job set $N = A \dot{\cup} B$, processing times $p_j \in \{0,1\}$, and certain precedence requirements $\kappa(b)$ for each job~$b \in B$.

In this section, we prove Theorem~\ref{thm:4approximation:allbutoneMSSC}.
That is, we give a $4$-approximation algorithm for the special case of GMSSC where $\kappa(b) = \max\{d(b) - 1,1\}$ with $d(b) := \abs{\mathcal{P}(b)}$ for all $b \in B$.
So each job in $B$ requires all but one of its predecessors to be completed before it can start, unless it has only one predecessor (\emph{all-but-one MSSC}).
Suppose we want to schedule a job $b \in B$ with $d(b) \geq 2$ at time $t \geq 0$.
Then we need at least $d(b) - 1$ of its predecessors to be completed before $t$.
Equivalently, for each pair of distinct $i,j \in \mathcal{P}(b)$ at most one of the two jobs $i,j$ may complete after $t$.
This gives the following linear relaxation with the same time-indexed variables as before and time horizon $T =\sum_{j \in N} p_j \leq n$.
\begin{subequations}\label{GMSSC:allbutone:LP}
\begin{align}
&\min& \sum\limits_{b \in B} \sum\limits_{t = 1}^{T} w_b \cdot t \cdot x_{bt} \\
& \text{s.t.} & \sum\limits_{t=1}^T x_{jt} &= 1  & \forall &\, j \in N, \label{GMSSC:allbutone:LP:execute} \\
 				&& \sum\limits_{j \in N} \sum\limits_{s=t - p_j +1}^{t} x_{js} &\leq 1  & \forall &\, t \in [T], \label{GMSSC:allbutone:LP:nooverlap}\\
 				&& \sum\limits_{s=1}^{t+p_b} x_{bs} - \sum\limits_{s=1}^{t} ( x_{is} + x_{js}) &\leq 0 &\forall &\, b \in B, \ \forall i,j \in \mathcal{P}(b), \ \forall \, t \in [T -p_b], \label{GMSSC:allbutone:LP:OR} \\
 				&& \sum\limits_{s=1}^{t+p_b} x_{bs} - \sum\limits_{s=1}^{t} x_{is} &\leq 0 & \forall &\, b \in B: \, \mathcal{P}(b) = \{i\}, \ \forall \, t \in [T -p_b], \label{GMSSC:allbutone:LP:AND} \\
 				&& x_{jt} 	&\geq 0 & \forall &\, j \in N, \ \forall \, t \in [T] \label{GMSSC:allbutone:LP:relax}.
\end{align} 
\end{subequations}
Constraints~(\ref{GMSSC:allbutone:LP:execute}) and~(\ref{GMSSC:allbutone:LP:nooverlap}) again ensure that each job is processed and no jobs overlap, respectively.
Note that only jobs with non-zero processing time contribute to~(\ref{GMSSC:allbutone:LP:nooverlap}).
If $d(b) = 1$, then~(\ref{GMSSC:allbutone:LP:AND}) dominates~(\ref{GMSSC:allbutone:LP:OR}).
It ensures that the unique predecessor of $b \in B$ is completed before $b$ starts.
Note that this is a classical AND-precedence constraint which will not affect the approximation factor.

If $d(b) \geq 2$, then~(\ref{GMSSC:allbutone:LP:OR}) models the above observation.
Suppose at most $d(b) - 2$ predecessors of $b$ complete before $t$.
Then there are $i,j \in \mathcal{P}(b)$ such that $\sum_{s=1}^t (x_{is} + x_{js}) = 0 \geq \sum_{s=1}^{t} x_{bs}$, so $b$ cannot complete by time $t$.

Note that we can solve LP~(\ref{GMSSC:allbutone:LP}) in polynomial time.
Similar to the algorithms of Section~\ref{sec:MSSC}, we first solve LP~(\ref{GMSSC:allbutone:LP}) and let $\overline{x}$ be an optimal fractional solution.
We then draw $\beta$ randomly from $(0,1]$ with density function $f(\beta) = 2\beta$, and schedule the jobs in $A$ and $B$ in order of non-decreasing $\frac{\beta}{2}$-points and $\beta$-points, respectively.
Again, we break ties consistently with precedence constraints.
(Choosing $\frac{\beta}{2}$ for the jobs in $A$ ensures that at most one of the predecessors of a job $b \in B$ is scheduled after $b$ in the constructed schedule.)
This algorithm is called Algorithm~3.

\begin{lemma}\label{lem:GMSSC:allbutone:approximation}
Algorithm 3 is a randomized 4-approximation for all-but-one MSSC.
\end{lemma}

\begin{proof}
The proof is fairly similar to the proof of Lemma~\ref{lem:ANDprecsbipartiteunitp}.
Let $\overline{x}$ be an optimal fractional solution to LP~(\ref{GMSSC:allbutone:LP}) and let $\overline{C}_j := \sum_t t \cdot \overline{x}_{jt}$ be the fractional completion time of job $j$.
Let $b \in B$ with $d(b) \geq 2$ and $0 < \beta \leq 1$ and $\alpha = \frac{\beta}{2}$.
Then at least $d(b) - 1$ predecessors of $b$ have their $\alpha$-point before $t^\beta_b$.
Suppose not, and let $i,j \in \mathcal{P}(b)$ such that $t^\alpha_i, t^\alpha_j > t^\beta_b$. Then
\begin{equation}
\sum_{s=1}^{t^\beta_b - p_b} (\overline{x}_{is} + \overline{x}_{js} ) \leq \sum_{s=1}^{t^\beta_b} \overline{x}_{is} + \sum_{s=1}^{t^\beta_b} \overline{x}_{js} < \alpha + \alpha = \beta \leq \sum_{s=1}^{t^\beta_b} \overline{x}_{bs},
\end{equation}
contradicts~(\ref{GMSSC:allbutone:LP:OR}).
So the schedule returned by Algorithm 3 is feasible.

Similar to~(\ref{ineq:integral:t}) in the proof of Lemma~\ref{lem:ANDprecsbipartiteunitp}, we observe that $\int_0^1 t^\beta_b d\beta \leq \overline{C}_b$.
Let $C_j(\beta)$ be the completion time of $j$ in the resulting schedule for a realization of $\beta$, and let $\prec$ be the order of the jobs in this schedule.
Observe that $C_b(\beta) = \sum_{i \preceq b} p_i \leq  \frac{t^\beta_b}{\alpha} = \frac{2}{\beta} t^\beta_b$, as in~(\ref{ineq:completion:t}).
If we draw $\beta$ randomly from $(0,1]$ with density function $f(\beta) = 2\beta$, the expected completion time of $b \in B$ is
\begin{equation}
\mathbb{E}[C_b(\beta)] = \int_0^1 f(\beta) \frac{2}{\beta} t^\beta_b d\beta = 4 \int_0^1 t^\beta_b d\beta \leq 4 \overline{C}_b.
\end{equation}
Since only jobs in $B$ contribute to the objective function this proves the claim.
\end{proof}

One can derandomize Algorithm 3 similar to Lemma~\ref{lem:ANDprecsbipartite:derandomize}, which proves Theorem~\ref{thm:4approximation:allbutoneMSSC}.
Note that Algorithm 3 also works if jobs in $B$ have unit processing time.
It can be generalized to release dates and arbitrary processing times, if we use an interval-indexed formulation similar to LP~(\ref{ANDprecsMSSC:generalprocessing:LP}).
If we choose $\varepsilon' = \frac{\varepsilon}{4}$ and solve the corresponding interval-indexed formulation instead of LP~(\ref{GMSSC:allbutone:LP}), then Algorithm 3 is a $(4 +\varepsilon)$-approximation for any $\varepsilon > 0$.
Again, AND-precedence constraints do not affect the approximation factor, similar to Lemma~\ref{lem:ANDprecsbipartiteunitp} and Lemma~\ref{lem:ANDprecsMSSC:generalprocessingtime}.
The following lemma shows that the analysis of Algorithm 3 is tight.

\begin{lemma}\label{lem:GMSSC:integralitygap}
The integrality gap of LP~\emph{(\ref{GMSSC:allbutone:LP})} is equal to 4.
\end{lemma}
\begin{proof}
Let $n \in \mathbb{N}$ be even, and let $A = \{a_1,\dots,a_n\}$ and $B = \{b_1,\dots,b_n\}$ with $\mathcal{P}(b_i) = A \setminus \{a_i\}$ and $\kappa(b_i) = n-2$ for all $i \in [n]$, see Figure~\ref{fig:GMSSC:integralitygap}.
Further $p_a = w_b = 1$ and $w_a = p_b = 0$ for all $a \in A$ and $b \in B$.
Note that this is an instance of all-but-one MSSC.

An optimal solution would schedule the jobs in $A$ in any arbitrary order, and then process each job in $B$ as early as possible.
W.l.og., we assume that the jobs in $A$ are scheduled such that job $a_i \in A$ completes at time $i$.
So jobs $b_{n-1},b_n \in B$ both can complete at time $n-2$, and all other jobs in $B$ complete at time $n-1$.
The objective value of this schedule is equal to $2 (n-2) + (n-2)(n-1) = (n-2) (n+1)$.

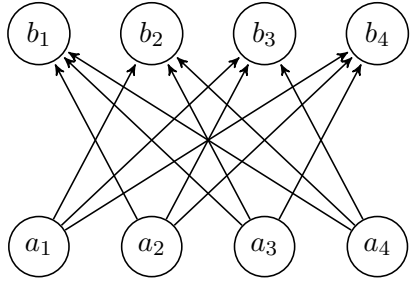
\begin{figure}
\centering
\begin{tikzpicture}[->,>=stealth',shorten >=1pt,auto,node distance=1.5cm,semithick]

		\node[circle, draw, minimum size=8mm] (a1) {$a_1$};
		\node[circle, draw, minimum size=8mm] (a2) [right of=a1] {$a_2$};
		\node[circle, draw, minimum size=8mm] (a3) [right of=a2] {$a_3$};
		\node[circle, draw, minimum size=8mm] (a4) [right of=a3] {$a_4$};
		
		\node[circle, draw, minimum size=8mm] (b1) [above=2cm of a1] {$b_1$};
		\node[circle, draw, minimum size=8mm] (b2) [right of=b1] {$b_2$};
		\node[circle, draw, minimum size=8mm] (b3) [right of=b2] {$b_3$};
		\node[circle, draw, minimum size=8mm] (b4) [right of=b3] {$b_4$};
		
		\draw (a1) -- (b2);
		\draw (a1) -- (b3);
		\draw (a1) -- (b4);
		
		\draw (a2) -- (b1);
		\draw (a2) -- (b3);
		\draw (a2) -- (b4);
				
		\draw (a3) -- (b1);
		\draw (a3) -- (b2);
		\draw (a3) -- (b4);
		
		\draw (a4) -- (b1);
		\draw (a4) -- (b2);
		\draw (a4) -- (b3);
			
\end{tikzpicture}
\caption{{\small Structure of the instances (here $n=4$) for which the integrality gap of LP~(\ref{GMSSC:allbutone:LP}) approaches 4.}}
\label{fig:GMSSC:integralitygap}
\end{figure}
Now consider the following fractional solution where $x_{at} = \frac{1}{n}$ for all $a \in A$ and $1 \leq t \leq n = T$.
For $b \in B$, set $x_{bt} = \frac{2}{n}$ for $1 \leq t \leq \frac{n}{2}$ and $x_{bt} = 0$ else.
One can easily verify that this solution is feasible for LP~(\ref{GMSSC:allbutone:LP}).
Its objective value is equal to
\begin{equation}
\sum_{b \in B} \sum_{t = 1}^n t \cdot x_{bt} = n \sum_{t=1}^{\frac{n}{2}} t \cdot \frac{2}{n} = \frac{n}{2}\left(\frac{n}{2} + 1 \right) = \frac{1}{4} n(n+2).
\end{equation}
So the integrality gap of LP~(\ref{GMSSC:allbutone:LP}) approaches~$4$ as $n$ goes to infinity.
\end{proof}

\section{Generalizing the Greedy Algorithm for Pipelined Set Cover}\label{sec:bipartite}

In this section, we generalize the greedy algorithm of Munagala et al.~\cite{Munagala2005} for pipelined set cover to a 4-approximation for $1 \, | \, or\text{-}prec=bipartite \, | \, \sum w_j C_j$.
Recall that pipelined set cover is a special case of $1 \, | \, or\text{-}prec=bipartite \, | \, \sum w_j C_j$, where $p_b = 0$ for all $b \in B$.
Prior to this, no constant-factor approximation was known for the general problem of scheduling with bipartite OR-precedence constraints.

For MSSC, the greedy algorithm always chooses the vertex with most uncovered hyperedges next~\cite{FeigeLovaszTetali2002,FeigeLovaszTetali2004}.
For pipelined set cover~\cite{Munagala2005}, it chooses the vertex that maximizes the ratio $\frac{w(B_v)}{c_v}$, where $B_v \subseteq \{e \in \mathcal{E} \, | \, v \in e\}$ is the set of uncovered hyperedges incident to $v$ and $w(B_v) := \sum_{e \in B_v} w_e$.
We can view $1 \, | \, or\text{-}prec=bipartite \, | \, \sum w_j C_j$ as if hyperedges, i.e.~jobs in $B$, are associated with positive processing times.
In this case, we successively schedule a $\rho$-maximizing feasible starting set, where $\rho(S) := \frac{w(S)}{p(S)}$.
We call $S \subseteq N$ a \textit{feasible starting set}, if we can schedule the jobs in $S$ without violating any OR-precedence constraints.
The set of feasible starting sets is denoted by $\mathcal{S}$.
That is, $S \in \mathcal{S}$, if $b \in B \cap S$ implies $\mathcal{P}(b) \cap S \not= \emptyset$.

Formally, the algorithm works as follows.
Let $\mathcal{A} = A \cup \{b \in B \, | \, \mathcal{P}(b) = \emptyset \}$ be the set of \emph{available jobs}, i.e.~jobs that can start in a feasible schedule. 
First, we compute a feasible starting set $S \in \mathcal{S}$ such that $\rho(S) \geq \rho(S')$ for all $S' \in \mathcal{S}$.
Then, we append the jobs in $S$ in any order at the end of the current schedule, and remove the jobs in $S$ from the instance.
Finally, we update the set of available jobs $\mathcal{A}$, and repeat.
Note that this algorithm generalizes the greedy algorithms of~\cite{FeigeLovaszTetali2002,FeigeLovaszTetali2004,Munagala2005}.
The technique is similar to results in expanding search~\cite{AlpernLidbetter2013,FokkinkLidbetterVegh2019}, and its analysis uses the same neat histogram argument as in~\cite{FeigeLovaszTetali2004}.

\begin{lemma}\label{lem:rhomaximizingset}
A $\rho$-maximizing feasible starting set can be computed in polynomial time.
\end{lemma}
\begin{proof}
For technical reasons, we define $\rho(\emptyset) := -1 < \rho(S)$ for any non-trivial $S \subseteq N$.
We define a set $S_j \in \mathcal{S}$ for every job $j \in \mathcal{A}$, and show that the set $\argmax_{j \in \mathcal{A}} \rho(S_j)$ is a $\rho$-maximizing set.
For $j \in B \cap \mathcal{A}$, let $S_j = \{j\}$.
For $j \in A$, let $S_j = \{j\} \cup B_j$ where $B_j = \{b^j_1,\dots,b^j_l\} \subseteq \{b \in B \, | \, j \in \mathcal{P}(b)\}$ is ordered such that $\rho(\{b^j_1\}) \geq \rho(\{b^j_2\}) \geq \cdots \geq \rho(\{b^j_l\}) \geq \rho(\{b\})$ for all $b \in \{B \setminus B_j \, | \, j \in \mathcal{P}(b)\}$, and $\rho(\{b^j_k\}) > \rho(\{j,b^j_1,\dots,b^j_{k-1}\})$ for all $k \in [l]$.
We can construct the sets $S_j$ for $j \in \mathcal{A}$ in polynomial time by successively adding a successor with highest $\rho$-value to $S_j$, if this increases the overall $\rho$-ratio.
Let $j^* \in \mathcal{A}$ such that $\rho(S_{j^*}) \geq \rho(S_j)$ for all $j \in \mathcal{A}$.

Let $S \in \mathcal{S}$ be an inclusion-minimal $\rho$-maximizing feasible starting, i.e.~$\rho(S) \geq \rho(S')$ for all $S' \in \mathcal{S}$.
We will show that $\rho(S) = \rho(S_{j^*})$.
Note that $\rho(\{b\}) \leq \rho(S)$ for every $b \notin S$ with $\mathcal{P}(b) \cap S \not= \emptyset$ or $b \in \mathcal{A}$, because otherwise we could add $b$ to $S$ and obtain $\rho(S \cup \{b\}) > \rho(S)$.
Since $S$ is a feasible starting set, $\mathcal{A} \cap S \not= \emptyset$.
We claim that $S_j \subseteq S$ for all $j \in \mathcal{A} \cap S$.
This is trivially true for any $j \in B \cap \mathcal{A} \cap S$.
So let $j \in A \cap S$ and suppose that $S_j \not\subseteq S$.
Let $k$ be minimal such that $b^j_k \in B_j \setminus S$. 
Then by the above observations and the construction of $S_j$, it holds $\rho(S) \geq \rho(\{b^j_k\}) > \rho(\{j,b^j_1,\dots,b^j_{k-1}\})$ and $\rho(S) \geq \rho(\{b^j_k\}) \geq \rho(\{b^j_i\}) \geq \rho(\{b\})$ for all $i \geq k$ and $b \in B \setminus B_j$ with $j \in \mathcal{P}(b)$.
But then $S \setminus (\{j\} \cup \{b \in B \, | \, j \in \mathcal{P}(b)\})$ is a feasible starting set with strictly higher $\rho$-value.
Hence, $S_j \subseteq S$ for all $j \in \mathcal{A} \cap S$.

Suppose that there is $j \in A \cap S$, and let $S= S_j \dot{\cup} \overline{S}$ where $\overline{S} \not= \emptyset$.
Note that $\rho(S) \geq \rho(S_j) \geq \rho(\{b\})$ for all $b \in \{B \setminus B_j \, | \, j \in \mathcal{P}(b)\}$ implies that $S \cap (B \setminus B_j) = \emptyset$, since $S$ was chosen to be an inclusion-minimal $\rho$-maximizing set.
Hence, $\overline{S} \in \mathcal{S}$.
Further, $\rho(S) \geq \rho(S_j)$ is equivalent to
\begin{align}\label{rho:equivalence}
\frac{w(S)}{p(S)} = \frac{w(S_j) + w(\overline{S})}{p(S_j) + p(\overline{S})} \geq \frac{w(S_j)}{p(S_j)} \ \Longleftrightarrow \  w(\overline{S}) p(S_j) \geq w(S_j) p(\overline{S}) \ \Longleftrightarrow \ \rho(\overline{S}) \geq \rho(S_j).
\end{align}
Similarly, $\rho(S) \geq \rho(\overline{S})$ implies $\rho(S_j) \geq \rho(\overline{S})$, so $\rho(S_j) = \rho(\overline{S})$.
But then a similar transformation as in~(\ref{rho:equivalence}) yields $\rho(S) = \rho(S_j)$, which contradicts the inclusion-minimality of $S$.
If $\mathcal{A} \cap S \subseteq B$, then $S \subseteq B$, so inclusion-minimality of $S$ implies that $S = \{j^*\}$.
\end{proof}

With Lemma~\ref{lem:rhomaximizingset} in place, it is obvious that the greedy algorithm runs in polynomial time, and that we can schedule the $\rho$-maximizing set optimally. It remains to show that this gives the desired approximation factor.
\begin{theorem}\label{thm:noAND:4approximation}
The greedy algorithm is a $4$-approximation for $1 \, | \, or\text{-}prec=bipartite\, | \, \sum w_j C_j$.
\end{theorem}
\begin{proof}
The proof is fairly similar to the histogram proof of Feige et al.~\cite{FeigeLovaszTetali2004}.
Suppose that the greedy algorithm terminates after $m$ stages, where we added a $\rho$-maximizing set to the end of the current schedule in each stage.
For $i \in [m]$, let $S_i$ be the set of jobs that are scheduled in stage $i$, and let $R_i := \bigcup_{l=i}^m S_l$ be the set of remaining jobs at the beginning of stage $i$.
That is, $\rho(S_i) \geq \rho(S)$ for all feasible starting sets $S$ of the remaining instance on~$R_i$.
We denote the completion time of job $j$ in the greedy schedule and an arbitrary, but fixed, optimal schedule by $C^G_j$ and $C^*_j$, respectively. 
Further, set $\phi_i := \frac{w(R_i)}{w(S_i)} p(S_i) = \frac{w(R_i)}{\rho(S_i)}$ and note that
\begin{align}\label{histogramGreedy}
\sum_{j \in N} w_j C^G_j \leq \sum_{l=1}^m \sum_{j \in S_l} w_j \sum_{i = 1}^{l} p(S_i) = \sum_{i = 1}^m p(S_i) \sum_{l = i}^m w(S_l) = \sum_{i = 1}^m p(S_i) w(R_i) = \sum_{i = 1}^m w(S_i) \phi_i.
\end{align}

We will construct two histograms that represent the objective values of the optimal solution and of the greedy solution, respectively.
We show that if we shrink the second one by a factor of 4, it fits into the first one.
This then yields the claim.
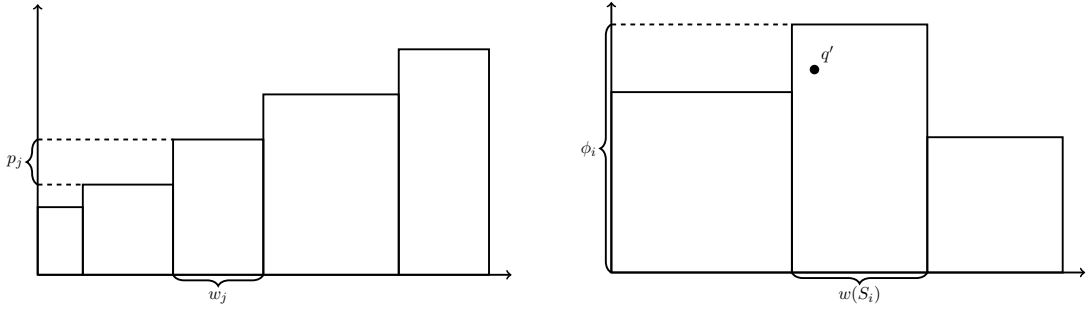
\begin{figure}
\centering
\scalebox{0.6}{
\begin{tikzpicture}[very thick,decoration={brace,amplitude=7pt,mirror}]
	\draw (0,0) rectangle (1,1.5);
	\draw (1,0) rectangle (3,2);
	\draw (3,0) rectangle (5,3);
	\draw (5,0) rectangle (8,4);
	\draw (8,0) rectangle (10,5);
	
	\draw[dashed] (0,3) -- (3,3);
	\draw[dashed] (0,2) -- (1,2);
	\draw[decorate] (0,3) -- (0,2) node[midway,xshift=-0.5cm] {$p_j$};
	\draw[decorate] (3,0) -- (5,0) node[midway,yshift=-0.5cm] {$w_j$};
	
	\draw[->] (0,0) -- (10.5,0);
	\draw[->] (0,0) -- (0,6);
\end{tikzpicture}}
\hspace*{0.5cm}
\scalebox{0.6}{
\begin{tikzpicture}[very thick,decoration={brace,amplitude=7pt,mirror}]
	\draw (0,0) rectangle (4,4);
	\draw (4,0) rectangle (7,5.5);
	\draw (7,0) rectangle (10,3);
	
	\draw[dashed] (0,5.5) -- (4,5.5);
	\draw[decorate] (0,5.5) -- (0,0) node[midway,xshift=-0.5cm] {$\phi_i$};
	\draw[decorate] (4,0) -- (7,0) node[midway,yshift=-0.5cm] {$w(S_i)$};
	
	\fill (4.5,4.5) circle (3pt) node[above right] {$q'$};
	
	\draw[->] (0,0) -- (10.5,0);
	\draw[->] (0,0) -- (0,6);
\end{tikzpicture}}
\caption{{\small Histogram corresponding to an optimal solution (left) and to the greedy solution (right).}}
\label{fig:histogram}
\end{figure}
The first histogram contains a column for each job $j \in N$ with width $w_j$ and height $C^*_j$ in the order the jobs appear in the optimal solution.
Note that the height of the columns is non-decreasing, and that the total area of the histogram is equal to $\sum_j w_j C^*_j$, see Figure~\ref{fig:histogram} (left).
The second histogram consists of $m$ columns, one for each stage, in the order the stages appear in the greedy schedule.
The width of column $i \in [m]$ is $w(S_i)$, and its height is equal to $\phi_i$, see Figure~\ref{fig:histogram} (right).
The total area of the second histogram is equal to $\sum_i w(S_i) \phi_i \geq \sum_j w_j C^G_j$, see~(\ref{histogramGreedy}).
\begin{figure}
\centering
\scalebox{0.6}{
\begin{tikzpicture}[very thick,decoration={brace,amplitude=7pt,mirror}]	
	\draw[gray] (0,0) rectangle (1,1.5);
	\draw[gray] (1,0) rectangle (3,2);;
	\draw[gray] (3,0) rectangle (5,3);
	\draw[gray] (5,0) rectangle (8,4);
	\draw[gray] (8,0) rectangle (10,5);
	
	\draw (5,0) rectangle (7,2);
	\draw (7,0) rectangle (8.5,2.75);
	\draw (8.5,0) rectangle (10,1.5);
	
	\fill (7.25,2.25) circle (3pt) node[above right] {$q$};
	
	\draw[->] (0,0) -- (10.5,0);
	\draw[->] (0,0) -- (0,6);
\end{tikzpicture}}
\caption{{\small Shrunk histogram (black) aligned right inside the first histogram (gray).}}
\label{fig:histogram:shrink}
\end{figure}
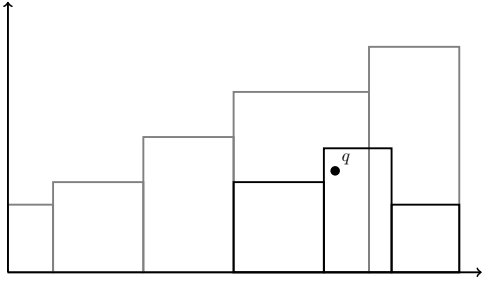

Now, we shrink the second histogram by a factor 2 in height, and a factor 2 in width, and align it to the right, see Figure~\ref{fig:histogram:shrink}.
So the total area of the shrunk histogram is equal to $\frac{1}{4} \sum_i w(S_i) \phi_i$.
We claim that each point of the shrunk histogram is contained in the first histogram.
This then implies that the area of the shrunk histogram is less or equal than $\sum_j w_j C^*_j$, which yields
\begin{equation}
\sum_j w_j C^G_j \leq \sum_i w(S_i) \phi_i \leq 4 \sum_j w_j C^*_j.
\end{equation}
To prove the claim, let $q'$ be a point in the second histogram, and suppose it is contained in column~$i$.
Let $q$ be the corresponding point in the shrunk histogram. 
So the height of $q$ is at most $\frac{1}{2} \phi_i= \frac{w(R_i)}{2 \rho(S_i)}$, and its distance to the right is at most $\frac{1}{2} \sum_{l = i}^m w(S_l) = \frac{1}{2} w(R_i)$.

Recall that $S_i$ satisfies $\rho(S_i) \geq \rho(S)$ for all feasible starting sets $S$ of the remaining instance on~$R_i$.
That is, no feasible schedule (even not the optimal one) can cover more than an amount of $\lambda \rho(S_i)$ of weight of the jobs in $R_i$ during $\lambda$ time units, even if it processes only jobs in $R_i$.
If the schedule processes also jobs that are not in $R_i$ during that time, it can process even less weight of $R_i$, since all processing times are non-negative.
Hence, within $\frac{1}{2} \phi_i$ time units, the optimal solution cannot cover more than an amount of $\frac{1}{2} \phi_i \rho(S_i)$ of weight of jobs in $R_i$.
So at time $\frac{1}{2} \phi_i$ there is at least an amount of
\begin{equation}
w(R_i) - \frac{1}{2} \phi_i \rho(S_i) = w(R_i) - \frac{1}{2} \frac{w(R_i)}{\rho(S_i)} \rho(S_i) = \frac{1}{2} w(R_i)
\end{equation}
of weight of jobs in $R_i$ unscheduled.
Thus the point $(\frac{\phi_i}{2},\frac{w(R_i)}{2})$ is contained in the first histogram. Note that $q$ is to the lower right of $(\frac{\phi_i}{2},\frac{w(R_i)}{2})$, so $q$ is contained in the first histogram. This proves the statement.
\end{proof}

Note that the histogram of~\cite{FeigeLovaszTetali2004} in Figure~\ref{fig:histogram} (left) is just the flipped two-dimensional Gantt chart of~\cite{EastmanEvenIsaacs1964}.
Feige, Lovász and Tetali~\cite{FeigeLovaszTetali2004} observed that the analysis of the greedy algorithm is tight for MSVC, which is a special case of $1 \, | \, or\text{-}prec=bipartite\, | \, \sum w_j C_j$.
Although successively scheduling a $\rho$-maximizing set is similar to Sidney's decomposition~\cite{Sidney1975}, we do not get optimality of the greedy algorithm, because the $\rho$-maximizing set may not be unique, see~\cite{FeigeLovaszTetali2004}.

\section{Integrality Gaps for Other LP Relaxations}\label{sec:otherformulations}

In this section, we analyze other standard linear programming relaxations that have been useful for various scheduling problems, and show that they fail in the presence of OR-precedence constraints.
More precisely, we show that the natural LPs in linear ordering variables (Section~\ref{sec:linearordering}) and completion time variables (Section~\ref{sec:completiontime}) both exhibit integrality gaps that are linear in the number of jobs, even on instances where $E_{\wedge} = \emptyset$ and $\Delta = 2$. 

\subsection{Linear Ordering Formulation}\label{sec:linearordering}

The following relaxation for single-machine scheduling problems was proposed by Potts~\cite{Potts1980}.
It is based on linear ordering variables $\delta_{ij}$, which indicate whether job $i$ precedes job $j$ ($\delta_{ij} = 1$) or not ($\delta_{ij} = 0$).
This LP has played an important role in better understanding Sidney's decomposition~\cite{Sidney1975,CorreaSchulz2005}, and in uncovering the connection between AND-scheduling and vertex cover~\cite{ChudakHochbaum1999,CorreaSchulz2005,AmbuhlMastrolilli2009,AmbuhlMastrolilliMutsanasSvensson2011}.
A nice feature of this formulation is that we can model OR-precedence constraints in a very intuitive way with constraints $\sum_{a \in \mathcal{P}(b)} \delta_{ab} \geq 1$ for all $b \in B$.
Together with the total ordering constraints ($\delta_{ij} + \delta_{ji} = 1$), standard transitivity constraints ($\delta_{ij} + \delta_{jk} + \delta_{ki} \geq 1$) and AND-precedence constraints ($\delta_{ij} = 1$) we thus obtain a polynomial size integer program for $1 \, | \, ao\text{-}prec=A \dot{\vee} B \, | \, \sum w_j C_j$.
The LP-relaxation is obtained by relaxing the integrality constraints to $\delta_{ij} \geq 0$.
\begin{subequations} \label{IPformulation:linearordering:LP}
\begin{align}
& \min & \sum\limits_{j \in N} \sum\limits_{i \in N} w_j p_i \delta_{ij}  \\
& \text{s.t.} &  \delta_{ij} + \delta_{ji} &= 1  &\forall & \, i,j \in N: \, i \not= j  \label{IPformulation:linearordering:linearorder}\\
 				&& \delta_{ij} + \delta_{jk} + \delta_{ki} &\geq 1  &\forall & \, i,j,k \in N \label{IPformulation:linearordering:transitivity} \\
 				&& \sum\limits_{a \in \mathcal{P}(b)} \delta_{ab} &\geq 1  &\forall & \, b \in B: \, \mathcal{P}(b) \not= \emptyset, \label{IPformulation:linearordering:OR} \\
 				&& \delta_{ij} 					&= 1 & \forall &\, (i,j) \in E_{\wedge}, \label{IPformulation:linearordering:AND} \\
 				&& \delta_{ii} 					&= 1 & \forall &\, i \in N, \label{IPformulation:linearordering:consistency} \\
 				&&  \delta_{ij}					&\geq 0 & \forall &\, i,j \in N.
\end{align}
\end{subequations}
We set $\delta_{ii}  = 1$ in~(\ref{IPformulation:linearordering:consistency}) so the completion time of job $j$ is $C_j = \sum_i p_i \delta_{ij}$.
Note that every feasible single-machine schedule without idle time corresponds to a feasible integer solution of LP~(\ref{IPformulation:linearordering:LP}), and vice versa.
If $E_{\vee} = \emptyset$, i.e., $\mathcal{P}(b) = \emptyset$ for all $b \in B$, then this relaxation has an integrality gap of 2 (lower and upper bound of 2 due to~\cite{ChekuriMotwani1999} and~\cite{Schulz1996}, respectively).
However, in the presence of OR-precedence constraints, the gap of LP~(\ref{IPformulation:linearordering:LP}) grows linearly in the number of jobs, even if $E_{\wedge} = \emptyset$ and $\Delta = 2$.

\begin{lemma}\label{lem:IPformulation:linearordering:gap}
There is a family of instances such that the integrality gap of LP \emph{(\ref{IPformulation:linearordering:LP})} is $\Omega(n)$.
\end{lemma}

\begin{proof}
Let $n$ be a multiple of 3.
Consider an instance that consists of $m = \frac{n}{3}$ copies of the following directed graph on three jobs $\{i,k,j\}$.
The processing times and weights are equal to $p_{i} = p_k = 1$, $p_j = 0$ and $w_i = w_k = 0$, $w_j = 1$.
The jobs $i,k$ do not have predecessors, and $\mathcal{P}(j) = \{i,k\}$.
We indicate the job sets of copy $q \in [m]$ by $N_q = \{i_q,k_q,j_q\}$, and set $N = N_1 \cup \dots \cup N_m$.
That is, $A = \{i_1,\dots,i_m,k_1,\dots,k_m\}$ and $B = \{j_1,\dots,j_m\}$.
Note that $\abs{N} = 3m = n$, $E_{\wedge} = \emptyset$ and $\Delta = 2$, see Figure~\ref{fig:IPformulation:linearordering:gap} for an example.

\begin{figure}
\centering
\begin{tikzpicture}[->,>=stealth',shorten >=1pt,auto,node distance=1.5cm,semithick]

	\node [circle, draw, minimum size=8mm] (i1) {$i_1$};
	\node [circle, draw, minimum size=8mm] (j1) [above right of=i1] {$j_1$};
	\node [circle, draw, minimum size=8mm] (k1) [below right of=j1] {$k_1$};
	
	\path (i1) edge (j1);
	\path (k1) edge (j1);
	
	\node [circle, draw, minimum size=8mm] (i2) [right of=k1]{$i_2$};
	\node [circle, draw, minimum size=8mm] (j2) [above right of=i2] {$j_2$};
	\node [circle, draw, minimum size=8mm] (k2) [below right of=j2] {$k_2$};
	
	\path (i2) edge (j2);
	\path (k2) edge (j2);
	
	\node (dots1) [right of=k2] {$\dots$};
	\node (dots2) [above=0.7cm of dots1] {$\dots$};
	
	\node [circle, draw, minimum size=8mm] (i3) [right of=dots1] {$i_m$};
	\node [circle, draw, minimum size=8mm] (j3) [above right of=i3] {$j_m$};
	\node [circle, draw, minimum size=8mm] (k3) [below right of=j3] {$k_m$};
	
	\path (i3) edge (j3);
	\path (k3) edge (j3);

\end{tikzpicture}
\caption{{\small Instance for which LP~(\ref{IPformulation:linearordering:LP}) exhibits an integrality gap that is linear in the number of jobs. The processing times and weights are $p_{j_q} = 0$, $p_{i_q} = p_{k_q} = 1$ and $w_{j_q} = 1$, $w_{i_q} = w_{k_q} = 0$ for all $q \in [m]$.}}
\label{fig:IPformulation:linearordering:gap}
\end{figure}
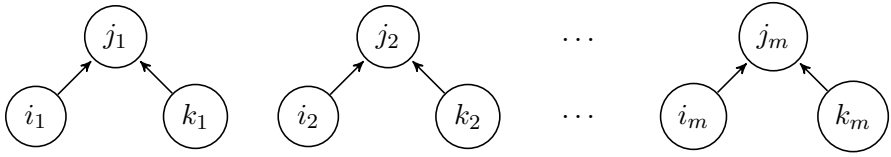

Any feasible schedule has to schedule $i_q$ or $k_q$ before $j_q$ for all $q \in [m]$.
Further, any optimal schedule would always schedule $j_q$ immediately after $i_q$ or $k_q$, whichever completes first.
Since $p_{i_q} = p_{k_q} = 1$ for all $q \in [m]$, it does not matter whether $i_q$ or $k_q$ precedes $j_q$, and the order of the copies does not matter either.
So the optimal integer solution has an objective value of $\sum_{q = 1}^m q = \frac{m(m+1)}{2} \in  \Omega(n^2)$.

Now consider the following fractional solution. For all $q \in [m]$ set $\delta_{i_q j_q} =  \delta_{k_q j_q} = \delta_{i_q k_q} = \frac{1}{2}$.
For distinct $q_1, q_2 \in [m]$, set $\delta_{i_{q_1} i_{q_2}} = \delta_{k_{q_1} k_{q_2}} = \delta_{j_{q_1} j_{q_2}} = \delta_{i_{q_1} k_{q_2}} = \frac{1}{2}$ and $\delta_{i_{q_1} j_{q_2}} = \delta_{k_{q_1} j_{q_2}} = 0$.
Further let $\delta_{ij} = 1 - \delta_{ji}$ and $\delta_{ii} = 1$ for all distinct $i,j \in N$.
Note that $\delta$ is a feasible solution of LP~(\ref{IPformulation:linearordering:LP}).

As for the objective value, recall that $p_{j_q} = w_{i_q} = w_{k_q} = 0$ and $p_{i_q} = p_{k_q}  = w_{j_q} = 1$ for all $q \in [m]$.
So the only variables that contribute to the objective function with a non-zero coefficient are $\delta_{i_{q_1} j_{q_2}}$ and $\delta_{k_{q_1} j_{q_2}}$ for all $q_1,q_2 \in [m]$.
The coefficient of these variables is equal to~$1$.
Further $\delta_{i_{q_1} j_{q_2}} = \delta_{k_{q_1} j_{q_2}} = 0$ for distinct $q_1,q_2 \in [m]$.
Hence the objective value of $\delta$ is equal to $\sum_{q_1 = 1}^m \sum_{q_2=1}^m \left(\delta_{i_{q_1} j_{q_2}} + \delta_{k_{q_1} j_{q_2}} \right) = \sum_{q=1}^m \left( \frac{1}{2} + \frac{1}{2} \right) = m \in \mathcal{O}(n)$.
Since $\delta$ is feasible, the optimal objective value of LP~(\ref{IPformulation:linearordering:LP}) is $\mathcal{O}(n)$.
Thus the integrality gap of LP~(\ref{IPformulation:linearordering:LP}) is $\Omega(n)$.
\end{proof}

Note that the instance in the proof of Lemma~\ref{lem:IPformulation:linearordering:gap} satisfies $\abs{\mathcal{P}(b)} \leq 2$ for all $b \in B$.
For this special case, we exhibit facet-defining inequalities in the remainder of this section.
If $\abs{\mathcal{P}(b)} \leq 2$ for all $b \in B$, then LP~(\ref{IPformulation:linearordering:LP}) can be written as

\begin{subequations} \label{IPformulation:linearordering:LP2}
\begin{align}
& \min & \sum\limits_{j \in N} \sum\limits_{i \in N} w_j p_i \delta_{ij}   \\
& \text{s.t.} & \delta_{ij} + \delta_{ji} &= 1 & \forall &\, i,j \in N: \, i \not= j ,\label{IPformulation:linearordering:LP2:linearorder}\\
 				&& \delta_{ij} + \delta_{jk} + \delta_{ki} &\geq 1  &\forall &\, i,j,k \in N, \label{IPformulation:linearordering:LP2:transitivity}\\
 				&& \delta_{ab} + \delta_{a'b} &\geq 1 & \forall&\, b \in B: \, \mathcal{P}(b)=\{a,a'\}, \label{IPformulation:linearordering:LP2:OR} \\
 				&& \delta_{ij} 					&= 1 & \forall &\, (i,j) \in E_{\wedge}, \text{ or } \mathcal{P}(j) = \{i\}, \label{IPformulation:linearordering:LP2:AND}\\
 				&& \delta_{ii}				&= 1 & \forall &\, i \in N, \label{IPformulation:linearordering:LP2:consistency}\\
 				&& \delta_{ij}					&\geq 0 & \forall &\, i,j \in N.
\end{align}
\end{subequations}

Note that constraints~(\ref{IPformulation:linearordering:LP2:linearorder}), (\ref{IPformulation:linearordering:LP2:transitivity}), and~(\ref{IPformulation:linearordering:LP2:consistency}) coincide with the corresponding constraints in LP~(\ref{IPformulation:linearordering:LP}).
Constraints~(\ref{IPformulation:linearordering:LP2:OR}) model the OR-precedence constraints for jobs $b \in B$ with $\abs{\mathcal{P}(b)} = 2$. For $b \in B$ with $\abs{\mathcal{P}(b)}= 1$, the corresponding OR-precedence constraint is equivalent to an AND-constraint and is included in~(\ref{IPformulation:linearordering:LP2:AND}).

\begin{theorem}\label{thm:ORfacetdefining}
For all $b \in B$ and $\mathcal{P}(b) = \{a,a'\}$, the constraints
\begin{equation}\label{IPformulation:linearordering:ineq:OR}
\delta_{aa'} + \delta_{a'b} \geq 1
\end{equation}
are valid for the integer hull of LP~\emph{(\ref{IPformulation:linearordering:LP2})}. Moreover, if they are tight, then they are either facet-defining or equality holds for all feasible integer solutions of LP~\emph{(\ref{IPformulation:linearordering:LP2})}.
\end{theorem}

It can easily be verified that the fractional solution in the proof of Lemma~\ref{lem:IPformulation:linearordering:gap} satisfies~(\ref{IPformulation:linearordering:ineq:OR}), and is feasible for LP~(\ref{IPformulation:linearordering:LP2}).
Hence the integrality gaps of LPs~(\ref{IPformulation:linearordering:LP}) and (\ref{IPformulation:linearordering:LP2}) remain linear even if we add constraints~(\ref{IPformulation:linearordering:ineq:OR}).
The proof of Theorem~\ref{thm:ORfacetdefining} is deferred to Appendix~\ref{appendix:linearordering}. 

Recall that the instance in the proof of Lemma \ref{lem:IPformulation:linearordering:gap} satisfies constraints~(\ref{IPformulation:linearordering:ineq:OR}), so the integrality gap remains linear, even with these additional constraints.
In GMSSC each job $b$ requires at least $\kappa(b) \in [\abs{\mathcal{P}(b)}]$ of its predecessors to be completed before it can start.
This can also be easily modeled with linear ordering variables by introducing a constraint $\sum_{a \in \mathcal{P}(b)} \delta_{ab} \geq \kappa(b)$.
However, note that the instance in the proof of Lemma~\ref{lem:IPformulation:linearordering:gap} is an instance of MSVC (which is a special case of MSSC and all-but-one MSSC).
So already for $\kappa(b) = 1$ or $\kappa(b) = \max\{\abs{\mathcal{P}(b)} -1,1\}$ and $\Delta = 2$ this formulation has an unbounded integrality gap.

\subsection{Completion Time Formulation}\label{sec:completiontime}

The LP relaxation examined in this section contains one variable $C_j$ for every job $j \in N$, which indicates the completion time of this job.
In the absence of precedence constraints, Wolsey~\cite{WolseyISMPTalk} and Queyranne~\cite{Queyranne1993} showed that the convex hull of all feasible completion time vectors can be fully described by the set of vectors $\{C \in \mathbb{R}^n \, | \, \sum_{j \in S} p_j C_j \geq f(S) \ \forall \, S \subseteq N\}$, where $f(S) := \frac{1}{2}\left(\sum_{j \in S} p_j \right)^2 + \frac{1}{2} \sum_{j \in S} p_j^2$ is a supermodular function.
One should note that, although there is an exponential number of constraints, one can separate them efficiently~\cite{Queyranne1993}.
In the presence of AND-precedence constraints, Schulz~\cite{Schulz1996} proposed the following 2-approximation algorithm.
The algorithm solves the corresponding linear program with additional constraints $C_j \geq C_i + p_j$ for $(i,j) \in E_{\wedge}$ and schedules the jobs in non-decreasing order of their LP-values.

For OR-precedence constraints, we use the concept of \emph{minimal chains}, see e.g.~\cite{Happach2019}, to generalize the parallel inequalities of~\cite{WolseyISMPTalk,Queyranne1993}.
More specifically, we present a class of inequalities that are valid for all feasible completion time vectors of an instance of $1 \, | \, ao\text{-}prec=A \dot{\vee} B \, | \, \sum w_j C_j$, and that, in the absence of precedence constraints, coincide with the parallel inequalities.
We add inequalities for AND-precedence constraints in the obvious way, $C_j \geq C_i + p_j$ for $(i,j) \in E_{\wedge}$, so we assume $E_{\wedge} = \emptyset$ for the moment.
Recall that $S \in \mathcal{S}$ is a feasible starting set if $j \in B \cap S$ implies $\mathcal{P}(j) \cap S \not= \emptyset$.
The length of a minimal chain of a job $k$ w.r.t.~a set $S \subseteq N$ is defined as

\begin{equation}\label{minchain}
mc(S,k) := \min\{ \sum_{j \in T} p_j \ | \ T \subseteq N: \exists \, U \subseteq S \cup T \text{ with } k \in U \in \mathcal{S}\}.
\end{equation}

Intuitively, the value $mc(S,k)$ is the minimal amount of time that we need to schedule job $k$ in a feasible way, if we can schedule the jobs in $S$ for free, i.e.~if we assume all jobs in $S$ have zero processing time.
We call $T \in \argmin(mc(S,k))$ a \textit{minimal chain of $k$ w.r.t.~$S$}.
If $k \notin S$ then $k \in T$ for all minimal chains $T$.
Let $2^N$ be the power set of $N$. For all $k \in N$, we define the set function
\begin{equation}\label{ineq:minchain:rhs}
f_k(S): 2^N \to \mathbb{R}_{\geq 0}, \quad f_k(S) := \frac{1}{2} \bigg(\sum_{j \in S} p_j + mc(S,k)\bigg)^2 + \frac{1}{2} \bigg(\sum_{j \in S} p_j^2 + mc(S,k)^2\bigg).
\end{equation}

Note that if $k \in S \in \mathcal{S}$, then $mc(S,k) = 0$, so $f_k(S) = f(S) = \frac{1}{2}\left(\sum_{j \in S} p_j \right)^2 + \frac{1}{2} \sum_{j \in S} p_j^2$.
In particular, (\ref{ineq:minchain:rhs}) generalizes the function $f: 2^N \to \mathbb{R}_{\geq 0}$ of~\cite{WolseyISMPTalk,Queyranne1993} to OR-precedence constraints. One can also show that $mc(\cdot,k)$ and $f_k(\cdot)$ are supermodular, for any $k$. The proof of the next theorem is deferred to Appendix~\ref{appendix:completiontime}.

\begin{theorem}\label{thm:singlemachinepolytope:validrelaxation}
For any $k \in N$ and $S \subseteq N$ the inequality 
\begin{equation}\label{ineq:minchain}
\sum\limits_{j \in S} p_j C_j + mc(S,k) \, C_k \geq f_k(S)
\end{equation}
is valid for all feasible completion time vectors. Moreover, if there is $T \in \argmin(mc(S,k))$ such that $S \cup T \in \mathcal{S}$ is a feasible starting set, then~\emph{(\ref{ineq:minchain})} is tight.
\end{theorem}

Theorem~\ref{thm:singlemachinepolytope:validrelaxation} suggests the following natural LP-relaxation for $1 \, | \, ao\text{-}prec=A \dot{\vee} B \, | \, \sum w_j C_j$:
\begin{subequations}\label{Completiontime:LP}
\begin{align}
& \min & \sum_{j \in N} w_j C_j & \\
& \text{s.t.} & \sum_{j \in S} p_j C_j + mc(S,k) \, C_k &\geq f_k(S) & \forall &\, k \in N, \ \forall \, S \subseteq N, \label{Completiontime:minchain}\\
				&& C_j - C_i &\geq p_j & \forall &\, (i,j) \in E_{\wedge} \label{Completiontime:AND}.
\end{align}
\end{subequations}

Note that it is not clear how to separate constraints~(\ref{Completiontime:minchain}) in polynomial time.
The gap of LP~(\ref{Completiontime:LP}) can grow linearly in the number of jobs, even for instances of $1 \, | \, ao\text{-}prec=A \dot{\vee} B \, | \, \sum w_j C_j$ with $E_{\wedge} = \emptyset$ and $\Delta = 2$.

\begin{lemma}\label{lem:generalminchain:gap}
There is a family of instances such that the gap between an optimal solution of LP~\emph{(\ref{Completiontime:LP})} and an optimal schedule is $\Omega(n)$.
\end{lemma}

\begin{proof}
Let $m \in \mathbb{N}$.
Consider the example depicted in Figure~\ref{fig:generalminchain:gap} with $n = 2m+1$ jobs and sets $A = \{a,i_1,\dots,i_m\}$ and $B=\{j_1,\dots,j_m\}$.
The processing times and weights are $p_a = \frac{m}{2}$, $w_a = 0$, and $p_{i_q} = w_{j_q} = 1$, $w_{i_q} = p_{j_q} = 0$ for all $q \in [m]$.
The predecessors of jobs in $B$ are $\mathcal{P}(j_q) = \{a,i_q\}$ for all $q \in [m]$.
It holds $mc(\emptyset,i_q) = mc(\emptyset,j_q) = 1 < \frac{m}{2} = mc(\emptyset,a) $ for all $q \in [m]$.
Note that $E_{\wedge} = \emptyset$ and $\Delta = 2$.

\begin{figure}
\centering
\begin{tikzpicture}[->,>=stealth',shorten >=1pt,auto,node distance=1.5cm,semithick]

	\node [circle, draw,minimum size=8mm] (1) {$j_1$};
	\node [circle, draw,minimum size=8mm] (2) [right of=1] {$j_2$};
	\node  (dots) [right of=2] {$\cdots$};
	\node [circle, draw,minimum size=8mm] (d) [right of=dots] {$j_m$};
	
	\node [circle, draw,minimum size=8mm] (1a) [below=1cm of 1] {$i_1$};
	\node [circle, draw,minimum size=8mm] (2a) [right of=1a] {$i_2$};
	\node  (dotsa) [right of=2a] {$\cdots$};
	\node [circle, draw,minimum size=8mm] (da) [right of=dotsa] {$i_m$};
	
	\node [circle, draw,minimum size=8mm] (i) [above=1cm of dots] {$a$};

	\draw (1a) edge (1);
	\draw (2a) edge (2);
	\draw (da) edge (d);
	\draw (i) edge (1);
	\draw (i) edge (2);
	\draw (i) edge (d);
\end{tikzpicture}
\caption{{\small Instance for which LP~(\ref{Completiontime:LP}) exhibits an integrality gap that is linear in the number of jobs. The processing times and weights are $p_a = \frac{m}{2}$, $w_a = w_{i_q} = p_{j_q} = 0$, and $w_{j_q} = p_{i_q} = 1$ for all $q \in [m]$.}}
\label{fig:generalminchain:gap}
\end{figure}
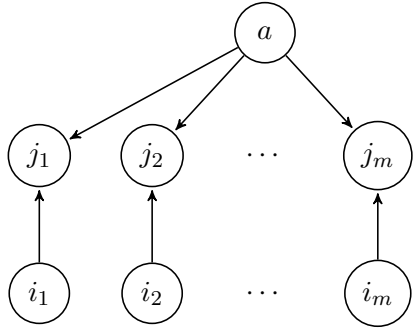

Due to the structure of the precedence relation, there are only two reasonable schedules that obey the precedence constraints.
Let $C'$ and $C''$ be the completion time vectors of the schedules $a \rightarrow \{j_1,\dots,j_m\} \rightarrow \{i_1,\dots,i_m\}$ and $\{i_q \rightarrow j_q\} \rightarrow a$, respectively.
The notion $\{i_q \rightarrow j_q\}$ indicates that we schedule pairs $i_q \rightarrow j_q$ for all $q \in [m]$ consecutively in arbitrary order.
One can easily verify that any other schedule has a strictly larger objective value than $C'$ or $C''$.
The objective function values of $C'$ and $C''$ are equal to $\frac{m}{2} \cdot m$ and $\sum_{q=1}^m q = \frac{m(m + 1)}{2}$, respectively.
Since $\frac{m^2}{2} < \frac{m(m+1)}{2}$, the optimal schedule is $C'$ with an objective value of $\Omega(n^2)$.
Now consider $C^*$ defined as $C^*_{j_q} = 1$, $C^*_{i_q} = q + 1$ for all $q \in [m]$ and $C^*_a = \frac{3}{2} m + 1$.
The objective function value of $C^*$ is equal to $m \in \mathcal{O}(n)$, so the gap of the objective function values of $C'$ and $C^*$ is $\Omega(n)$.

It remains to show that $C^*$ is feasible for LP~(\ref{Completiontime:LP}), i.e.~it satisfies constraints~(\ref{Completiontime:minchain}).
Note that $C^*$ corresponds to a schedule with idle time on the single machine (i.e.~no jobs overlap) of the following form: $idle \rightarrow \{j_1,\dots,j_m\} \rightarrow \{i_1,\dots,i_m\} \rightarrow idle \rightarrow a$.
Hence it satisfies the constraints $\sum_{l \in S} p_l \, C^*_l \geq f(S)$ for all $S \subseteq N$ of~\cite{WolseyISMPTalk,Queyranne1993}.
We now show that also $\sum_{l \in S} p_l \, C^*_l + mc(S,k) C^*_k \geq f_k(S)$ for all $k \in N$ and all $S \subseteq N$.
 
First let $k \in A$. The value of $mc(S,k)$ is then either equal to 0 (if $k \in S$) or $p_k$ (if $k \notin S$). In either case we obtain $f_k(S) = \frac{1}{2}\left( \sum_{l \in S \cup \{k\}} p_l \right)^2 + \frac{1}{2}\left(\sum_{l \in S \cup \{k\}} p_l^2 \right) = f(S \cup \{k\})$
and thus
\begin{equation}
 \sum_{l \in S} p_l \, C^*_l + mc(S,k) \, C^*_k = \sum_{l \in S \cup \{k\}} p_l \, C^*_l \geq f(S \cup \{k\}) = f_k (S).
\end{equation}

Now let $k = j_{q} \in B$, set $t:= \abs{S \cap \{i_1,\dots,i_m\}} \leq m$.
First suppose that $a \in S$, so $mc(S,k) = 0$. It holds
$f_k(S) = \frac{1}{2}\left(\frac{m}{2} + t \right)^2 + \frac{1}{2} \left(\frac{m^2}{4} + t \right) = \frac{m^2}{4} + \frac{1}{2} \, t^2 + \frac{m+1}{2} \, t$.
We obtain
\begin{align}
\begin{split}
\sum_{l \in S} p_l \, C^*_l + mc(S,k)C^*_k &= p_a \, C^*_a + \sum_{i \in S \cap \{i_1,\dots,i_m\}} p_{i} \, C^*_{i} \geq \frac{m}{2}\left(\frac{3}{2}m+1\right) + \sum_{q=1}^t (q+1) \geq \\
&\geq \frac{m^2}{4} + \left(\frac{m}{2} + 1\right) t + \frac{t(t+1)}{2} \geq \frac{m^2}{4} +  \frac{m+1}{2} t + \frac{1}{2} t^2 = f_k(S).
\end{split}
\end{align}
The first inequality holds with equality if the $t$ jobs in $S\cap \{i_1,\dots,i_m\}$ are those with lowest indices, otherwise it is strict.
For the second inequality we use $m+1 \geq t$.
If $a \notin S$, and $i_{q} \in S$, we get $mc(S,k) =0$ and $f_k(S) = \frac{1}{2}t^2 + \frac{1}{2} t = \frac{t(t+1)}{2}$. Similar to before, we obtain
\begin{equation}
\sum_{l \in S} p_l \, C^*_l + mc(S,k) \, C^*_k =\sum_{i \in S \cap \{i_1,\dots,i_m\}} p_{i} \, C^*_{i} \geq \sum_{q=1}^t (q+1) \geq \frac{t(t+1)}{2} = f_k(S).
\end{equation}
Finally, if $\mathcal{P}(k) \cap S = \emptyset$, then $mc(S,k) = 1$. So $f_k(S) = \frac{1}{2}(t + 1)^2 + \frac{1}{2} (t + 1) = \frac{t+1}{2} (t+2)$, and
\begin{align}
\begin{split}
\sum_{l \in S} p_l \, C^*_l + mc(S,k) \, C^*_k &= \sum_{i \in S \cap \{i_1,\dots,i_m\}} p_{i} \, C^*_{i} + C^*_k \geq \sum_{q=1}^t (q+1) + 1 =\\
&=  \frac{t(t+1)}{2} + t + 1 = \frac{t+1}{2} (t + 2) = f_k(S).
\end{split}
\end{align}
So $C^*$ satisfies constraints~(\ref{Completiontime:minchain}) for all $k \in N$ and $S \subseteq N$, and is feasible for LP~(\ref{Completiontime:LP}).
\end{proof}

\section{NP-Hardness of Restricted Special Cases}\label{sec:NPhardness}

As already indicated in the introduction, $1 \, | \, ao\text{-}prec=A \dot{\vee} B \, | \, \sum w_j C_j$ generalizes several NP-hard problems (see Figure~\ref{fig:overview}), so it is certainly NP-hard.
Theorem~\ref{thm:sumofcompletion:nphard} strengthens the NP-hardness result of Johannes~\cite{Johannes2005} for scheduling OR-precedence constrained jobs with unit processing times.

\begin{theorem}\label{thm:sumofcompletion:nphard}
$1 \, | \, or\text{-}prec=bipartite, \, p_j \in \{0,1\} \, | \, \sum C_j$ and $1 \, | \, or\text{-}prec=bipartite, \, p_j = 1 \, | \, \sum w_j C_j$ with $w_j \in \{0,1\}$ are strongly NP-hard.
\end{theorem}

\begin{proof}
The reduction goes from Exact 3-Set Cover which is known to be strongly NP-hard~\cite{GareyJohnson1979}.
The input of an Exact 3-Set Cover instance consists of a positive integer $q$, a universe $U = \{e_1,\dots,e_{3q}\}$ and a collection of subsets $\mathcal{R}$ of $U$ where each $S \in \mathcal{R}$ is of size $\abs{S} = 3$.
The task is to decide whether or not there is an exact cover for $U$, i.e.~$\mathcal{T} \subseteq \mathcal{R}$ with $\abs{\mathcal{T}} = q$ such that $U = \bigcup_{S \in \mathcal{T}} S$.

Let $(q,U,\mathcal{R})$ be an instance of Exact 3-Set Cover.
We introduce one job for every set in $\mathcal{R}$ (set-jobs) and one for every element of $U$ (element-jobs).
The graph representing the precedence constraints is $G=(\mathcal{R} \cup U,E_{\vee})$ with $E_{\vee} =\{(S,e) \in \mathcal{R} \times U \, | \, e \in S\}$.
Note that there are no edges within jobs in $\mathcal{R}$ or within jobs in $U$ and that the out-degree of each set-job in $G$ is equal to three.
Hence, this is an instance of $1 \, | \, or\text{-}prec=bipartite \, | \, \sum w_j C_j$.
The weights and processing times depend on the initial scheduling problem:
\begin{enumerate}[(i)]
\item\label{unitw} for $1 \, | \,  or\text{-}prec=bipartite, \, p_j \in \{0,1\} \, | \, \sum C_j$, set $p_{S} = 1$ for all $S \in \mathcal{R}$ and $p_{e} = 0$ for all~$e \in U$.
\item\label{unitp} for $1 \, | \,  or\text{-}prec=bipartite, \, p_j = 1 \, | \, \sum w_j C_j$, set $w_{S} = 0$ for all $S \in \mathcal{R}$ and $w_{e} = 1$ for all $e \in U$.
\end{enumerate}
Note that an optimal schedule for any of the two problems will schedule all successors of a set-job immediately, since the element-jobs are the only ones with zero processing time or positive weight, respectively.

Suppose that $(q,U,\mathcal{R})$ is a YES-instance and let $\mathcal{T}= \{S_1,\dots,S_q\} \subseteq \mathcal{R}$ be an exact cover for~$U$.
So it is feasible to first schedule $S_1$ followed by the three elements it contains, then $S_2$ followed by the three elements it contains, and so on.
With $\abs{\mathcal{R}} = m$ the objective value of this schedule for~(\ref{unitw}) is equal to
\begin{equation}
\sum_{S \in \mathcal{R}} C_{S} + \sum_{j=1}^q C_{S_j} \abs{S_j} = \sum_{j=1}^m j + 3 \sum_{j=1}^q j = \frac{m(m+1)}{2} + \frac{3q(q+1)}{2} .
\end{equation}
For~(\ref{unitp}) the objective value is equal to
\begin{equation}
\sum_{j=1}^q \sum_{e \in S_j} C_{e} = \sum_{j =1}^q (4j - 2) + \sum_{j=1}^q (4j - 1) + \sum_{j=1}^q 4j =12 \frac{q(q+1)}{2} - 3 q = 6q^2+ 3q.
\end{equation}

These are in fact the lowest possible objective values, since every set-job \emph{activates} at most three element-jobs.
So each schedule with these objective values starts with $q$ contiguous and disjoint ``blocks'' of the form $(S_j,e_{1_j},e_{2_j},e_{3_j}) \in \mathcal{R} \times U^3$.
The corresponding $q$ set-jobs then form an exact cover for~$U$. 
If we could solve $1 \, | \, or\text{-}prec=bipartite, \, p_j \in \{0,1\} \, | \, \sum C_j$ or $1 \, | \,  or\text{-}prec=bipartite, \, p_j = 1 \, | \, \sum w_j C_j$ with $w_j \in \{0,1\}$ in polynomial time, then we could decide whether or not $(q,U,\mathcal{R})$ is a YES-instance of Exact 3-Set Cover.
\end{proof}

Note that the second problem in Theorem~\ref{thm:sumofcompletion:nphard} is a special case of the problem considered in~\cite{Johannes2005}, and that $1 \, | \, ao\text{-}prec=A \dot{\vee} B, \, p_j = 1 \, | \, \sum C_j$ is trivial.

\section{Conclusion}\label{sec:remarks}

In this paper, we analyze single-machine scheduling problems with certain AND/OR-precedence constraints that are extensions of min-sum set cover, precedence-constrained min-sum set cover, pipelined set cover, minimum latency set cover, and set cover.
Using machinery from the scheduling context, we derive new approximation algorithms for the general problem that rely on solving time-indexed linear programming relaxations and scheduling jobs according to random $\alpha$-points.
In a nutshell, one may say that the new key technique is to choose the value of $\alpha$ for jobs in $A$ dependent on the corresponding $\beta$-value of jobs in $B$.

This observation allows us also to derive the best constant-factor approximation algorithm for an interesting special case of the generalized min-sum set cover problem---the all-but-one MSSC problem---which in itself is a generalization of min-sum vertex cover.
This $4$-approximation algorithm may further support the conjecture of Im et al.~\cite{ImSviridenkoZwaan2014}, namely that GMSSC is $4$-approximable.
Further, we present the first constant-factor approximation for scheduling jobs subject to bipartite OR-precedence constraints, which generalizes the previously-known greedy algorithms for min-sum set cover and pipelined set cover.

It is easy to see that one can also include AND-precedence constraints between jobs in $A$ and $B$, i.e., allow $E_{\wedge} \subseteq (A \times N) \cup (B \times B)$.
This does not affect the approximation guarantees or feasibility of the constructed schedules, since $\alpha \leq \beta$ and constraints~(\ref{ANDprecsbipartiteunitp:LP:AND}) imply $t^\alpha_a \leq t^\beta_b$ for $(a,b) \in E_{\wedge}$.
Similarly, $l^\alpha_a \leq l^\beta_b$ for $(a,b) \in E_{\wedge}$ follows from~(\ref{ANDprecsMSSC:generalprocessing:LP:AND}).
Note that it is not clear whether the analyses of the algorithms in Section~\ref{sec:MSSC} are tight.

Besides deriving approximation algorithms based on time-indexed LPs, we analyze other standard LP relaxations, namely linear ordering and completion time formulations.
These relaxations facilitated research on scheduling with AND-precedence constraints, see e.g.~\cite{Potts1980,Queyranne1993,Schulz1996,HallSchulz1997,ChudakHochbaum1999,CorreaSchulz2005,AmbuhlMastrolilli2009,AmbuhlMastrolilliMutsanasSvensson2011}.
For the integer hull of the linear ordering relaxation we present a class of facet-defining valid inequalities and we generalize the well-known inequalities of~\cite{WolseyISMPTalk,Queyranne1993} for the completion time relaxation.
We show that, despite these additional constraints, both relaxations exhibit linear integrality gaps, even if $\Delta = 2$ and $E_{\wedge} = \emptyset$.
Thus, unless one identifies stronger valid inequalities, these formulations seem to fail as soon as OR-precedence constraints are incorporated.

The results in Section~\ref{sec:otherformulations} mostly also apply to arbitrary OR-networks.
One can show that constraints~(\ref{IPformulation:linearordering:ineq:OR}) are facet-defining or implicit inequalities if the graph $G$ is acyclic w.r.t.~$E_{\vee}$.
The  validity of constraints~(\ref{ineq:minchain}) transfers to general OR-networks.
However, the functions $mc(\cdot,k)$ and $f_k(\cdot)$ are not supermodular if the graph contains several OR-layers.
In view of the integrality gaps in Sections~\ref{sec:GMSSC} and~\ref{sec:otherformulations}, it would be interesting to obtain stronger bounds on the integrality gap of the time-indexed formulation considered in Section~\ref{sec:MSSC}.

\paragraph*{Acknowledgments.} The authors would like to thank Thomas Lidbetter for fruitful discussions and for suggesting the algorithm in Lemma~\ref{lem:rhomaximizingset}.

\printbibliography

\newpage
\appendix

\section{Proofs for Section~\ref{sec:MSSC}}\label{appendix:MSSC}

\begin{proof}[Proof of Lemma~\ref{lem:ANDprecsMSSC:generalprocessingtime}]
Note that $\varepsilon' = \frac{\varepsilon}{2\Delta}$ is polynomial in the input.
So we can solve LP~(\ref{ANDprecsMSSC:generalprocessing:LP}) in polynomial time and thus Algorithm 2 runs in polynomial time.
Let $\overline{x}$ be an optimal solution of LP~(\ref{ANDprecsMSSC:generalprocessing:LP}) and let $\overline{C}_j = \sum_l \tau_{l-1} \, \overline{x}_{jl}$ be the fractional completion time of $j \in N$.

Let $0 < \beta \leq 1$. Recall the defintion of $\Delta = \max_{b \in B} \abs{\mathcal{P}(b)}$ and set $\alpha = \alpha(\beta) = \frac{\beta}{\Delta}$.
For any $(i,j) \in E_{\wedge}$, observe that $l^{\alpha}_{i} \leq l^{\alpha}_j$ (for $i,j \in A$) and $l^{\beta}_{i} \leq l^{\beta}_j$ (for $i,j \in B$) due to~(\ref{ANDprecsMSSC:generalprocessing:LP:AND}).
Constraint~(\ref{ANDprecsMSSC:generalprocessing:LP:OR}) implies that for every $b \in B$ with $\mathcal{P}(b) \not= \emptyset$ there is $a_b \in \mathcal{P}(b)$ that satisfies $\sum_{s=1}^{l^\beta_b} \overline{x}_{a_b s} \geq \frac{\beta}{\abs{\mathcal{P}(b)}} \geq \alpha$. So $l^{\alpha}_{a_b} \leq l^\beta_b$.
Hence the ordering the jobs according to $\prec$ respects all precedence constraints due to the tie breaking rules.
So the schedule returned by Algorithm 2 is feasible.

As for the approximation factor, fix $j \in N$.
For $l \in \{0,\dots,L\}$ let $\alpha_l = \sum_{k=1}^l \overline{x}_{jk}$ be the fraction of job $j$ that is completed until time $\tau_{l}$. Note that $\alpha_0 = 0$, $\alpha_L = 1$ and $\tau_{l^\gamma_j -1} \leq \tau_{l-1}$ for $\gamma \leq \alpha_l$.
We obtain
\begin{align}\label{ineq:integral:tau}
\begin{split}
\int_0^1 \tau_{l^\gamma_j - 1} d\gamma &= \sum_{l=1}^{L} \int_{\alpha_{l-1}}^{\alpha_{l}} \tau_{l^\gamma_j -1} d \gamma \leq \sum_{l=1}^{L} (\alpha_{l} - \alpha_{l-1}) \tau_{l-1} =\\
&= \sum_{l=1}^{L} \left( \sum_{k=1}^{l} \overline{x}_{jk} - \sum_{k=1}^{l-1} \overline{x}_{jk} \right) \tau_{l-1} = \sum_{l=1}^L \overline{x}_{jl} \cdot \tau_{l-1} = \overline{C}_j.
\end{split}
\end{align}
For $i, j \in N$ and $0 < \beta \leq 1$, let $\eta_i^j(\beta) = \sum_{k = 1}^{l^\alpha_j} \overline{x}_{ik}$ be the fraction of $i$ that is processed by time $\tau_{l^\beta_j}$.
Let $b \in B$ and $i \prec b$. Then $\alpha \leq \eta^b_i (\beta)$ (if $i \in A$) and $\alpha \leq \beta \leq \eta^b_i(\beta)$ (if $i \in B$), respectively.
Further~(\ref{ANDprecsMSSC:generalprocessing:LP:nooverlap}) implies
\begin{equation}\label{ineq:completion:tau}
\alpha \sum_{i \preceq b} p_i \leq \sum_{i \preceq b} \eta^b_i(\beta) \, p_i = \sum_{i \preceq b} \sum_{k = 1}^{l^\beta_b} p_i \, \overline{x}_{ik} \leq \sum_{i \in N} \sum_{k=1}^{l^\beta_b} p_i \, \overline{x}_{ik} \leq \tau_{l^\beta_b}
\end{equation}
Let $C_j(\beta)$ be the completion time of job $j$ in the schedule returned by Algorithm 2 for a realization of $\beta$.
It holds $C_b(\beta) \leq \sum_{i \preceq b} p_i$ for all $b \in B$.
So $C_b(\beta) \leq \sum_{i \preceq b} p_i \leq \frac{1}{\alpha}\tau_{l^\beta_{b}}$ by~(\ref{ineq:completion:tau}).
If we draw $\beta$ randomly from $(0,1]$ with density function $f(\beta) = 2\beta$, then the expected completion time of $b \in B$ is 
\begin{align}\label{ineq:expectation:tau}
\begin{split}
 \mathbb{E}[C_b(\beta)] &\leq \int_0^1 f(\beta) \frac{\Delta}{\beta} \tau_{l^\beta_b} d \beta =\int_0^1 f(\beta) \frac{\Delta}{\beta} (1+\varepsilon') \tau_{l^\beta_b -1} d \beta =\\
 &= 2 \Delta(1+\varepsilon') \int_0^1 \tau_{l^\beta_b -1} d \beta \leq (2 \Delta +\varepsilon) \overline{C}_b,
 \end{split}
\end{align}
where the last inequality is due to~(\ref{ineq:integral:tau}) and the choice of $\varepsilon'$.
Since only jobs in $B$ contribute to the objective value this proves the claim.
\end{proof}

\begin{proof}[Proof of Lemma~\ref{lem:ANDprecsMSSC:releasedates}]
The proof is inspired by~\cite{HallShmoysWein1996}.
We formulate the proof only for the case of arbitrary processing times.
It can be easily adapted to the 0/1 case, by informally replacing $\tau_l$ and $l$ by $t$ and setting $\varepsilon' = 0$.
In particular, ordering the jobs according to $\prec$ respects all precedence constraints.
If we add idle time where necessary such that each job starts only after its release date, the schedule returned by the algorithm is feasible.

As for the approximation factor, assume that there are jobs with non-trivial release dates, and let $0 < \beta \leq 1$.
Let $b \in B$ and $i \in A$ with $i \prec b$.
Then $\beta > 0$ implies $\alpha = \frac{\beta}{\Delta} >0$, and thus $\tau_{l^\beta_b} \geq \tau_{l^\alpha_i }\geq r_i$.
Similarly $\tau_{l^\beta_b} \geq \tau_{l^\beta_i} \geq r_i$ for $i \in B$ and $i \prec b$.
Hence $\tau_{l^\beta_b} \geq \max_{i \preceq b} r_i$.
The completion time of $b \in B$ in the schedule returned by the algorithm for a realization of $\beta$ is
\begin{equation}
C_b(\beta) \leq \max_{i \preceq b} r_i + \sum_{i \preceq b} p_i \leq \tau_{l^\beta_b} + \frac{1}{\alpha} \tau_{l^\beta_b} = \left( 1 + \frac{\Delta}{\beta} \right) \tau_{l^\beta_b} \leq \frac{\Delta + 1}{\beta} \tau_{l^\beta_b},
\end{equation}
where the second inequality follows from~(\ref{ineq:completion:tau}).
Note that~(\ref{ineq:completion:tau}) is not affected by release dates, since we only bound $\alpha \leq \eta^b_i(\beta) = \sum_{k = 1}^{l^\beta_b} \overline{x}_{ik}$ and use constraint~(\ref{ANDprecsMSSC:generalprocessing:LP:nooverlap}).
If we choose $\varepsilon' = \frac{\varepsilon}{2\Delta +2}$, then similar to~(\ref{ineq:expectation:tau}), the expected value of the completion time of $b \in B$ is
\begin{equation}
\mathbb{E}[C_b(\beta)] \leq \int_0^1 f(\beta) \frac{1 + \Delta}{\beta} \tau_{l^\beta_b} d\beta = 2(\Delta + 1) (1 + \varepsilon') \int_0^1 \tau_{l^\beta_b -1} d \beta \leq (2\Delta + 2 + \varepsilon) \overline{C}_b.
\end{equation}
Note that only jobs in $B$ contribute to the objective value. We can derandomize similar to Lemma~\ref{lem:ANDprecsbipartite:derandomize}, since each job only gets preempted at most once per time slot/interval also in the presence of release dates. This proves the claim.
\end{proof}

\section{Proofs for Section~\ref{sec:otherformulations}}

\subsection{Proof of Theorem~\ref{thm:ORfacetdefining}}\label{appendix:linearordering}

In the following, we interchangeably use $\delta$ to denote a total order of the jobs, i.e.~a single-machine schedule, and the corresponding 0/1 vector.
First, we discuss why constraints~(\ref{IPformulation:linearordering:ineq:OR}) are valid for any feasible schedule.
Note that any schedule (whether it is feasible or not) violates at most one of the constraints $\delta_{aa'} + \delta_{a'b} \geq 1$ or $\delta_{a'a} + \delta_{ab} \geq 1$, since $\delta_{aa'} + \delta_{a'a} = 1$ by~(\ref{IPformulation:linearordering:LP2:linearorder}).
Hence, in order for one of these inequalities to be violated, we need $\delta_{ab} = \delta_{a'b} = 0$.
But then $b$ precedes $a$ and $a'$, so the precedence constraints of $b$ are violated, and the schedule is infeasible.
Note that constraints~(\ref{IPformulation:linearordering:ineq:OR}) together with~(\ref{IPformulation:linearordering:LP2:linearorder}) dominate constraints~(\ref{IPformulation:linearordering:LP2:OR}).

To prove the second part of Theorem~\ref{thm:ORfacetdefining}, we make use of the following polyhedral observation.
Let $Q$ be the integer hull of the feasible region of LP~(\ref{IPformulation:linearordering:LP2}) if we drop constraints~(\ref{IPformulation:linearordering:LP2:AND}), i.e.~remove all AND-precedence constraints and OR-precedence constraints with only one OR-predecessor from $G$. 
That is, $Q := \conv(\{\delta \in \{0,1\}^{n^2} \, | \, \text{(\ref{IPformulation:linearordering:LP2:linearorder}),(\ref{IPformulation:linearordering:LP2:transitivity}),(\ref{IPformulation:linearordering:LP2:OR}),(\ref{IPformulation:linearordering:LP2:consistency})}\})$.
The precedence graph of the resulting instance of $1 \, | \, ao\text{-}prec=A \dot{\vee} B \, | \, \sum w_j C_j$ with $E_{\wedge} = \emptyset$ is denoted by $G'$.
Note that this instance satisfies $\abs{\mathcal{P}(b)} \in \{0,2\}$ for all $b \in B$.

Clearly, all feasible vectors $\delta \in Q$ satisfy $0 \leq \delta_{kl} \leq 1$ for all $k,l \in N$.
That is, for all distinct $k,l \in N$ the removed constraint of~(\ref{IPformulation:linearordering:LP2:AND}) defines a supporting hyperplane, call it~$H_{kl}$, at~$Q$.
In particular, for any facet $F$ of $Q$, either $F \cap H_{kl} \in \{\emptyset, Q \cap H_{kl}\}$ or $F \cap H_{kl}$ is a facet of $Q \cap H_{kl}$.
So in order to prove Theorem~\ref{thm:ORfacetdefining} it suffices to show that constraints~(\ref{IPformulation:linearordering:ineq:OR}) are facet-defining for $Q$.
We will do so by exhibiting $\dim(Q)$ affinely independent feasible vectors of $Q$ that satisfy~(\ref{IPformulation:linearordering:ineq:OR}) with equality.
Similar to~\cite{GrotschelJungerReinelt1985,QueyranneSchulz1994}, it is easy to see that $Q$ is not contained in any lower dimensional affine subspace than the one spanned by constraints~(\ref{IPformulation:linearordering:LP2:linearorder}) and~(\ref{IPformulation:linearordering:LP2:consistency}). So the dimension of $Q$ is equal to $d := \frac{n(n-1)}{2}$.

\begin{lemma} \label{lem:ORfacetdefining}
Constraints~\emph{(\ref{IPformulation:linearordering:ineq:OR})} are facet-defining for $Q$.
\end{lemma}
\begin{proof}
Let $b \in B$ with $\mathcal{P}(b) = \{a,a'\}$.
We prove the statement by exhibiting $d = \frac{n(n-1)}{2}$ affinely independent integer feasible points for $Q$ that satisfy $\delta_{aa'} + \delta_{a'b} = 1$.
Recall that the instance on $G'$ that corresponds to $Q$ satisfies $E_{\wedge} = \emptyset$ and $\abs{\mathcal{P}(b')} \in \{0,2\}$ for all $b' \in B$.
The proof goes by induction on the number of jobs~$n$.

The base case is $n =3$, i.e.~$d = 3$. There is only one possible graph $G'$ that can occur, see Figure \ref{fig:proof:ORfacetdefining} (left).
All feasible schedules for $G'$ are $a \rightarrow a' \rightarrow b$, $a' \rightarrow a \rightarrow b$, $a \rightarrow b \rightarrow a'$ and $a' \rightarrow b \rightarrow a$.
Obviously, all but the first schedule, $a \rightarrow a' \rightarrow b$, satisfy $\delta_{aa'} + \delta_{a'b} = 1$, and their respective $\delta$-vectors are affinely independent.
So the claim holds for $n=3$.

\begin{figure}
\centering
\begin{tikzpicture}[->,>=stealth',shorten >=1pt,auto,node distance=1.5cm,semithick]

		\node[circle, draw, minimum size=8mm] (i) {$a$};
		\node[circle, draw, minimum size=8mm] (j) [above right of=i] {$b$};
		\node[circle, draw, minimum size=8mm] (i1) [below right of=j] {$a'$};
				
		\path (i) edge (j);
		\path (i1) edge (j);

\end{tikzpicture}
\hspace*{1.3cm}
\begin{tikzpicture}[->,>=stealth',shorten >=1pt,auto,node distance=1.5cm,semithick]

		\node[circle, draw, minimum size=8mm] (i) {$a$};
		\node[circle, draw, minimum size=8mm] (j) [above right of=i] {$b$};
		\node[circle, draw, minimum size=8mm] (i1) [below right of=j] {$a'$};
		\node[circle, draw, minimum size=8mm] (l) [above right of=i1] {$j$};
		\node[circle, draw, minimum size=8mm] (k) [below right of=l] {$i$};

		\path (i) edge (j);
		\path (i1) edge (j);
		\path (i1) edge (l);
		\path (k) edge (l);

\end{tikzpicture}
\hspace*{1.3cm}
\begin{tikzpicture}[->,>=stealth',shorten >=1pt,auto,node distance=1.5cm,semithick]

		\node[circle, draw, minimum size=8mm] (i) {$a$};
		\node[circle, draw, minimum size=8mm] (j) [above right of=i] {$b$};
		\node[circle, draw, minimum size=8mm] (i1) [below right of=j] {$a'$};
		\node[circle, draw, minimum size=8mm] (l) [right of=i1] {$j$};
				
		\path (i) edge (j);
		\path (i1) edge (j);
		\path (l) edge[bend right] (j);
		
		\node[red,draw,very thick,cross out, inner sep=7pt,yshift=-0.1cm] (cross) [above left of=l] {};

\end{tikzpicture}
\caption{{\small Graphs for $n = 3$ (left), and for $n \geq 4$ with $B \setminus \{b\} \not= \emptyset$ (middle) and $B \setminus \{b\} = \emptyset$ (right). The crossed out arc cannot occur, since $\abs{\mathcal{P}(b)} = 2$ by assumption.}}
\label{fig:proof:ORfacetdefining}
\end{figure}
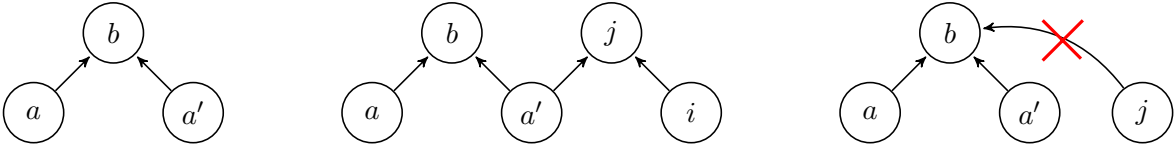

By induction hypothesis, we may assume that $\delta_{aa'} + \delta_{a'b} \geq 1$ is facet-defining for all instances on $n-1 \geq 3$ jobs with $\abs{\mathcal{P}(b')} \in \{0,2\}$ for all $b' \in B$ and $E_{\wedge} = \emptyset$.
Now consider an instance on $n \geq 4$ jobs with $\abs{\mathcal{P}(b')} \in \{0,2\}$ for all $b' \in B$ and $E_{\wedge} = \emptyset$.
We will remove a job $j \in N \setminus \{a,a',b\}$ from the instance in such a way that we can apply the induction hypothesis.
Then, we construct affinely independent feasible vectors based on the set of affinely independent vectors from the instance on $n-1$ jobs.
Feasibility of $\delta$ for $Q$ will follow from feasibility of the constructed schedule.

Note that the job $j$ can be chosen to have no successor.
Either $B \setminus \{b\} \not= \emptyset$ or, if $B \setminus \{b\} = \emptyset$, then there is a job in $A \setminus \{a,a'\}$ without successors, since $\abs{\mathcal{P}(b)} = 2$, see Figure~\ref{fig:proof:ORfacetdefining} (middle and right).
Hence we can choose $j \in B \setminus \{b\}$, or $j \in A \setminus \{a,a'\}$ has neither predecessors nor successors.
If we remove $j$ (and all arcs that end in $j$) from the instance, we are left with an instance on $n-1$ jobs.
By our choice of $j$, this instance satisfies $\abs{\mathcal{P}(b')} \in \{0,2\}$ for all $b' \in B \setminus \{j\}$.

So the induction hypothesis applies and there is a set $D'$ of $d' := \frac{(n-1)(n-2)}{2}$ affinely independent vectors that satisfy $\delta_{aa'} + \delta_{a'b} = 1$.
Note that these vectors correspond to feasible schedules on the instance without $j$.
If we add $j$ again, we add $n-1$ variables, which we will index by $\delta_{ij}$ for $i \in N \setminus \{j\}$.
Suppose that the last $n-1$ coordinates of the vectors correspond to these entries.
We show how to extend the affinely independent vectors in $D'$ to $d = d' + (n-1)$ affinely independent vectors $D$ in higher dimensional space.
For the sake of simplicity, we will omit ``transpose'' and assume that all vectors are column vectors.

First assume that $j$ was chosen to have neither predecessors nor successors.
Note that the vectors $D_1 := \{(\delta,0,\dots,0) \ | \ \delta \in D'\}$ are feasible ($j$ is scheduled first), affinely independent and satisfy~(\ref{IPformulation:linearordering:ineq:OR}) with equality.
It holds $\abs{D_1} = \abs{D'} = d'$.
Let $(\overline{\delta},0,\dots,0) \in D_1$ be a schedule where $j$ starts first.
We can successively move $j$ ``to the back'' of this schedule without loosing feasibility.
Thereby we obtain a set of vectors $D_2 := \{(\overline{\delta},1,0,\dots,0),(\overline{\delta},1,1,0,\dots,0),\dots,(\overline{\delta},1,\dots,1,0),(\overline{\delta},1,\dots,1)\}$ (up to permutation of the last $n-1$ coordinates).
Note that the components that appear in~(\ref{IPformulation:linearordering:ineq:OR}) are not changed, so all vectors in $D_2$ satisfy~(\ref{IPformulation:linearordering:ineq:OR}) with equality.
Obviously, $D_1 \cup D_2$ are affinely independent and $\abs{D_1 \cup D_2} = d' + n-1 = d$, which proves the claim.

Now assume that $j \in B \setminus \{b\}$, so it might not be feasible to schedule $j$ first, and we cannot move $j$ through the schedule as before.
Consider a schedule $\overline{\delta}$ that schedules the jobs in order $a' \rightarrow b \rightarrow a \rightarrow (A \setminus \{a,a'\}) \rightarrow (B \setminus \{j,b\})$, where the sets $A \setminus \{a,a'\}$ and $B \setminus \{j,b\}$ are scheduled in any arbitrary order.
Clearly, $\overline{\delta}$ is a feasible schedule for the instance on $n-1$ jobs and satisfies $\overline{\delta}_{aa'} + \overline{\delta}_{a'b} = 1$.
Hence, we may, w.l.o.g., assume that $\overline{\delta} \in D'$.
Define $\delta^j := (\overline{\delta},1,\dots,1)$ to be the schedule that first schedules $N \setminus \{j\}$ according to $\overline{\delta}$ and then $j$ last.
Note that the vectors in $D_0 := \{(\delta,1,\dots,1) \ | \ \delta \in D'\} \ni \delta^j$ are feasible ($j$ is scheduled last), affinely independent and satisfy the constraint with equality by induction hypothesis. 
It holds $\abs{D_0} = \abs{D'} = d'$. We construct a set $D_3$ of $\abs{D_3} = n-1$ vectors such that $D_0 \cup D_3$ are affinely independent as follows.

For every $i \in N \setminus \{j,a'\}$ let $\delta^i$ be the schedule that orders all jobs according to $\delta^j$, but shifts $i$ to the back of the schedule.
That is, $\delta^i$ swaps the order of $i$ and the set of jobs that appear in $\delta^j$ after $i$.
So in particular $\delta^i_{ij} = 0$.
Further, define $\delta^{a'}$ to be the schedule that orders the jobs $a \rightarrow b \rightarrow (A \setminus \{a,a'\}) \rightarrow (B \setminus \{b\}) \rightarrow j \rightarrow a'$. So, compared to $\delta^j$, job $a'$ is moved to the back and the order of $b$ and $a$ is reversed (this is crucial to maintain feasibility).
Set $D_3 := \{\delta^i \, | \, i \in N \setminus \{j\}\}$ with $\abs{D_3} = n-1$, and note that any $\delta^i \in D_3$ is feasible, since no job in $B$ has exactly one predecessor.

For $i \not= a'$ we did not swap the order of $a'$ and $\{b,a\}$ compared to $\overline{\delta}$, so $\delta^i$ satisfies~(\ref{IPformulation:linearordering:ineq:OR}) with equality.
For $i = a'$, it holds $\delta^{a'}_{aa'} + \delta^{a'}_{a'b} = 1 + 0 = 1$.
Further for $i \in N \setminus \{j\}$, $\delta^k_{ij} = 0$ iff $k = i$ for all $\delta^k \in D_3 $ and $\delta_{ij} = 1$ for all $\delta \in D_0$.
So $D_0 \cup D_3$ are $d' + n-1 = d$ affinely independent feasible vectors that satisfy~(\ref{IPformulation:linearordering:ineq:OR}) with equality.
This proves the claim.
\end{proof}

Lemma~\ref{lem:ORfacetdefining} together with the discussion before proves Theorem~\ref{thm:ORfacetdefining}.

\subsection{Proof of Theorem~\ref{thm:singlemachinepolytope:validrelaxation}}\label{appendix:completiontime}

\begin{proof}[Proof of Theorem~\ref{thm:singlemachinepolytope:validrelaxation}]
For $k \notin S$, it holds $mc(S,k) = p_k + mc(S \cup \{k\},k)$, and thus $f_k(S) \geq f_k(S \cup \{k\})$.
Note that the left-hand sides of~(\ref{ineq:minchain}) for $S$ and $S \cup \{k\}$ coincide in this case.
So for $k \in S$, inequality~(\ref{ineq:minchain}) is dominated by the corresponding constraint for $k$ and $S \setminus \{k\}$.
Further for $k \in S \in \mathcal{S}$, it holds $mc(S,k) = 0$, so~(\ref{ineq:minchain}) is equivalent to the inequality of~\cite{WolseyISMPTalk,Queyranne1993}:
\begin{equation}
 \sum_{j \in S} p_j C_j = \sum_{j \in S} p_j C_j + mc(S,k) C_k \geq f_k(S) = \frac{1}{2}\bigg( \sum_{j \in S} p_j \bigg)^2 + \frac{1}{2} \bigg( \sum_{j \in S} p_j^2 \bigg) = f(S).
\end{equation}
Note that all unit vectors have positive scalar product with the left-hand side of~(\ref{ineq:minchain}), so idle time in a schedule only increases the left-hand side of~(\ref{ineq:minchain}).
Hence, it suffices to show that all completion time vectors of schedules without idle time satisfy~(\ref{ineq:minchain}) for $k \in N$ and $S \subseteq N \setminus \{k\}$ with $S \notin \mathcal{S}$.

Fix $k \in N$ and $S \subseteq N \setminus \{k\}$ such that $S \notin \mathcal{S}$ is not a feasible starting set.
Let $T \in \argmin(mc(S,k))$ be a minimal chain and note that $k \in T$ and $\abs{T} \leq 2$.
We partition $S = S_1 \cup S_2 \cup S_3$, where $S_1 \in \mathcal{S}$ is an inclusion-maximal feasible starting set, and $S_2$ is inclusion-maximal such that $S_1 \cup T \cup S_2 \in \mathcal{S}$. All remaining jobs are contained in $S_3$.
The assumption $S \notin \mathcal{S}$ implies $S_2 \cup S_3 \not= \emptyset$.

First suppose that $S_3 = \emptyset$, i.e.~$S \cup T \in \mathcal{S}$.
Consider a feasible schedule that orders the jobs $S_1 \rightarrow T \rightarrow S_2 \rightarrow N \setminus (S \cup T)$.
Note that such a schedule exists, if we order the jobs within the sets suitably.
We will show that this schedule minimizes $\sum_{j \in N} w_j C_j$ with weights $w_j = p_j$ (for $j \in S$), $w_k = mc(S,k)$ and $w_j = 0$ (for $j \notin S$).
One can verify that its objective function value w.r.t.~these weights is indeed equal to~$f_k(S)$.

All jobs in $N \setminus (S \cup T)$ do not contribute to the objective function, so the objective value would only increase if we scheduled them earlier.
By Smith's rule~\cite{Smith1956}, permuting $T$ and the jobs $j \in S$ does not change the objective value (ratio $w_j/p_j = 1$ for all $j \in S_1 \cup S_2 \cup \{T\}$).
Here, we interpret $T$ as a single job, since, if $\abs{T} = 2$, we cannot schedule $k$ before its predecessor in $T$.
Since $S_1 \in \mathcal{S}$ was chosen to be inclusion-maximal, jobs in $S_2$ have to be scheduled after some job in $T$.
If $T = \{k\}$ this already proves that the schedule described above is optimal w.r.t.~the objective function.
Its objective value equals $f_k(S)$, so inequality~(\ref{ineq:minchain}) is valid and tight.

Now assume that $T = \{i,k\}$, i.e.~$k \in B$ and $i \in \mathcal{P}(k)$.
Note that $i \in \mathcal{P}(j)$ for any $j \in S_2$ by inclusion-maximality of $S_1$.
Suppose we move a job $j \in S_2$ between $i$ and $k$ (which would be feasible).
By Smith's rule, we can assume that $j$ is the job that directly succeeds $k$ in the schedule.
Its completion time decreases by $p_k \leq mc(S,k)$, whereas the completion time of $k$ increases by $p_j$.
Hence, the change in the objective value is equal to $mc(S,k) \, p_j - p_j \, p_k \geq 0$.
So the objective value is smaller, if $T$ precedes all jobs in $S_2$.
If we schedule a job $j \in S_1$ between $i$ and $k$, this only increases the objective value, since $C_j$ is increased, but $C_k$ and all completion times of jobs in $S_2$ remain.
Finally, $i$ was chosen to have minimal processing time such that $S_1 \cup T$ is a feasible starting set by definition of a minimal chain~(\ref{minchain}).
So altering $T$, e.g.~by exchanging $i$ with some other job in $\mathcal{P}(k)$, cannot decrease the completion time of $k$.
This proves that the schedule depicted above is optimal w.r.t.~the objective function. Since its objective value equals $f_k(S)$, inequality~(\ref{ineq:minchain}) is valid.
Further, equality holds for this particular schedule, so~(\ref{ineq:minchain}) is tight.

Now suppose that $S_3 \not= \emptyset$, i.e.~$T \cup S \notin \mathcal{S}$. 
Consider a schedule $S_1 \rightarrow T \rightarrow S_2 \rightarrow S_3 \rightarrow N \setminus (S \cup T)$, where the sets $S_1 \cup T \cup S_2$ are scheduled in a feasible way.
Define the weights as above, i.e.~set $w_j = p_j$ for $j \in S$, $w_k = mc(S,k)$ and $w_j = 0$ for $j \notin S$.
Note that the objective value $\sum_j w_j C_j$ of this schedule equals $f_k(S)$, but the schedule is not feasible due to $S_3$.
(If the schedule was feasible, then $S \cup T \in \mathcal{S}$ is a contradiction to the inclusion-maximality of $S_1$ or $S_2$.)
Similar to above, permuting jobs in $S \cup \{T\}$ does not change the objective value by Smith's rule~\cite{Smith1956}.
Also $S_1 \cup T \cup S_2$ have to appear in this block order, otherwise this would increase the objective function, see above.
Moreover, any feasible schedule has to schedule jobs in $N \setminus (S \cup T)$ earlier, to obey the precedence constraints of jobs in $S_3$.
Hence, to obtain a feasible schedule, we need to increase the completion time of jobs in $S_3$, which also increases the objective value.
So any feasible schedule has an objective value strictly greater than $f_k(S)$, and thus~(\ref{ineq:minchain}) is valid.
\end{proof}

\end{document}